\documentclass{article}

\usepackage[inline]{enumitem}
\RequirePackage[colorlinks,citecolor=blue,urlcolor=blue]{hyperref}
\usepackage[
    top=1in,
    bottom=1in,
    left=1in,
    right=1in
]{geometry}
\usepackage{amsthm}
\usepackage[noblocks]{authblk}
\usepackage{float}
\usepackage{amsmath} 
\usepackage{amsfonts}
\usepackage{caption}
\usepackage{multirow}
\usepackage{graphicx} 
\usepackage{booktabs} 
\usepackage{siunitx}  
\usepackage{natbib}
\usepackage{bm}

\newtheorem{theorem}{Theorem}
\newcommand{\tr}{{\rm tr}}

\newtheorem{proposition}[theorem]{Proposition}%

\newcommand{\E}{{\rm E}}
\newcommand{\Var}{{\rm Var}}
\newcommand{\Cor}{{\rm Cor}}
\newcommand{\Cov}{{\rm Cov}}

\newcommand{\x}{\boldsymbol{x}}
\newcommand{\w}{\boldsymbol{w}}
\newcommand{\bbeta}{\boldsymbol{\beta}}
\newcommand{\bgamma}{\boldsymbol{\gamma}}

\newcommand{\diag}{{\rm diag}}

\usepackage{color, xcolor}
\usepackage{soul} 

\begin{document}

\title{Principal Components Decomposition of Fraction of Variance Explained in High Dimensional Linear Models with Strong Correlation}

\author[1]{Man Luo}
\author[2,3]{Chun Chieh Fan}
\author[4]{David Azriel}
\author[*1,5]{Armin Schwartzman}

\affil[1]{Division of Biostatistics and Bioinformatics, Herbert Wertheim School of Public Health and Human Longevity Science, University of California San Diego, La Jolla, CA, USA}

\affil[2]{Center for Population Neuroscience and Genetics, Laureate Institute for Brain Research, Tulsa, OK, USA
}
\affil[3]{Department of Radiology, School of Medicine, University of California San Diego, La Jolla, CA, US}

\affil[4]{Technion–Israel institute of Technology, Haifa, Israel}

\affil[5]{Halıcıoğlu Data Science Institute, University of California San Diego, La Jolla, CA, US}

\affil[*]{Corresponding author. E-mail: armins@ucsd.edu }

\maketitle

\abstract{
The fraction of variance explained (FVE) in a linear model quantifies the extent to which predictors account for outcome variability. In high-dimensional settings, where traditional FVE estimators do not apply, modern FVE estimators such as GWASH or linear mix-effect model estimated through the restricted maximum likelihood (LMM-REML) struggle with strong correlation among predictors, often found, for example, in brain imaging data. We propose a decomposition framework that partitions the FVE into two components: a low-dimensional component capturing the strong correlation, estimable by low dimensional methods, and a high-dimensional component with remaining weak correlation, estimable by high dimensional methods. Simulations demonstrate that decomposing dominant principal components (PCs) and estimating the high-dimensional FVE using GWASH or LMM-REML leads to improved bias reduction compared to directly applying standard approaches such as GWASH and LMM-REML. Our method shows consistent performance asymptotically as both the number of predictors and the number of samples increase. We illustrate the method in an analysis of the Adolescent Brain Cognitive Development (ABCD) brain imaging dataset, capturing nuanced heritability signals in the FVE of cognitive measures predicted by high-resolution brain imaging data.}

\flushbottom
\thispagestyle{empty}

\section{Introduction}
The fraction of variance explained (FVE) serves as a measure of how much information about an outcome is accounted for by predictor variables. Its estimation has been significantly advanced by the emergence of high-dimensional data, particularly in genetics and neuroimaging. The availability of large-scale brain imaging datasets with thousands of features and genome-wide association studies (GWAS), which include millions of single nucleotide polymorphisms (SNPs), has opened new avenues for estimating FVE. Accurate FVE estimation can reveal the extent to which various factors affect phenotypes and improve our understanding of diseases. In neuroimaging, for instance, FVE estimation can quantify the extent to which brain imaging features explain inter-individual differences in neuroanatomy linked to mental illness, thereby providing insights into underlying neurobiological processes. This paper focuses on estimating FVE in neuroimaging studies involving scenarios with a single univariate phenotype, such as a quantitative trait.

Estimating FVE from neuroimaging data presents two main challenges. The primary challenge is the small sample size relative to the large number of features. In brain imaging, the relationship between the number of features ($m$) and the number of individuals ($n$) varies by modality. Voxel-based, surface-based and connectome-wide analyses often involve far more features than subjects ($m \gg n$), even in large datasets such as the UK Biobank. For example, voxel-based analyses (e.g. structural MRI, fMRI, DTI) typically involve $m$ ranging from 100{,}000 to over 1{,}000{,}000, depending on image resolution and field of view. Although sample sizes have increased with initiatives such as the Human Connectome Project (approximately 1{,}200, \cite{van_essen_wu-minn_2013}), the UK Biobank (approximately 500{,}000 individuals, \cite{sudlow_uk_2015}), and with more diverse genetic ancestry samples in the ABCD Study (approximately 12{,}000, \cite{casey_adolescent_2018}), $m$ almost always substantially exceeds $n$. In surface-based analyses (e.g. sMRI, fMRI), $m$ ranges from 10{,}000 to over 600{,}000 vertices across both hemispheres, depending on mesh density and reconstruction methods. Region of interest (ROI) analyses within connectome-wide frameworks may sometimes yield $n > m$, but $m$ can still be large. When $n > m$, FVE can be easily assessed using the adjusted $R^2$, a classical method for low-dimensional settings. However, in high-dimensional settings, standard linear regression using ordinary least squares is no longer applicable.\par The second challenge arises from the underlying biological correlation structure. In brain imaging, it is common for a small number of PCs to capture a large portion of the variance in the explanatory variables (\cite{palmer_distinct_2020}). This indicates a strong correlation structure between the original voxel or vertex predictors. 
While the assumptions of weak correlation, which underlie many existing methods (\cite{sabuncu_morphometricity_2016};\cite{yang_gcta_2011}), have a significant impact on the consistency of FVE estimators.

Within the neuroimaging field the specific FVE attributable to brain morphology is often termed as ``morphometricity" \citep{sabuncu_morphometricity_2016}. Inspired by the term ``heritability" which measures the proportion of variance of a trait explained by genetics, morphometricity quantifies the proportion of variance of a trait explained by brain morphology. The brain morphology can include whole-brain measurements or targeted subsets such as specific regions of interest (ROIs), hemispheres with features like vertex-wise cortical and subcortical features. Conceptually, morphometricity aims to capture the total contribution of brain measurements, offering a comprehensive view in high-dimensional settings. However, accurately estimating it in the face of these previously detailed challenges remains a key area of investigation. In this paper, we adopt the broader term FVE to refer to this estimand in neuroimaging studies. 

To estimate the FVE, linear mixed-effects models (LMM) have been used, with their parameters typically estimated using maximum likelihood or restricted maximum likelihood (REML) (\cite{sabuncu_morphometricity_2016}). This LMM-REML approach is inspired by its use in genetics to estimate SNP-based heritability, notably within the Genome-wide Complex Trait Analysis (GCTA) framework \citep{yang_gcta_2011}. This success has encouraged us to adapt various heritability estimation methods from SNP data for brain image analysis. Unlike genetic studies, where data access can be limited by privacy concerns, brain imaging data are often more accessible. This accessibility allows for applying a wider range of algorithms to neuroimaging data, including those requiring individual-level data. As a result, in this paper, we explore algorithms based on both the method of moments (MoM) and restricted maximum likelihood (REML) in neuroimaging; MoM can often use summary statistics, while REML typically requires individual-level data (\cite{zhu_statistical_2020, tang_review_2022}). For instance, our investigation considers MoM approaches such as the GWAS heritability (GWASH) estimator (\cite{schwartzman_simple_2019}), alongside the LMM-REML approach.

Although much effort has been devoted to estimating FVE when $m \gg n$, complications from strong correlations among predictors have received far less attention. Strong correlation is particularly prominent in neuroimaging data, where spatially adjacent and far away features may be highly dependent. For example, in T1-weighted structural MRI-derived surface area data, this dependency can manifest as a data matrix with a dominant eigenvector (see Figure \ref{Figure1_intro}). Such strong correlations can substantially impact the consistency and bias of FVE estimators (\cite{krishna_kumar_limitations_2016, couvyduchesne_unified_2020}). While some methods account for local correlation patterns like linkage disequilibrium (LD) in genetics (\cite{the_ucleb_consortium_reevaluation_2017,hou_accurate_2019}), they do not adequately address the types of strong, diffuse correlation structures prominent in neuroimaging. At the summary statistics level, GWASH was proved to provide unbiased FVE estimates but requires a weak correlation assumption among covariates \citep{schwartzman_simple_2019, azriel_consistency_2025}. We suspect this weak correlation assumption may be necessary for LMM-REML too. While a common strategy in neuroimaging known as ``demeaning"—centering data by subtracting the mean—can remove global effects or subject-level biases, it does not address these underlying strong correlation structures among predictors (\cite{kriegeskorte_representational_2008, azriel_estimation_2020,  revsine_unifying_2024}). 
 
In this paper, we show that FVE estimates can be biased when predictors are strongly correlated because the weak-correlation assumption is violated. Since these strong effects are often captured by the leading PCs of the predictor correlation matrix, we propose an FVE decomposition framework where predictors are projected onto the first few PCs and the total FVE is estimated as a combination of the low-dimensional contribution from the leading components and the high-dimensional contribution from the remaining weakly correlated residual features. To justify the validity of the high-dimensional portion of the estimate, we show both theoretically and via simulations that the predictor covariance matrix exhibits weak correlation after removing dominant PCs representing strong diffuse effects. 

Because population eigenvectors are unavailable in practice and estimating eigenvectors and FVE from the same dataset can cause bias due to double use of the data, we propose to estimate the leading eigenvectors from surrogate data, that is, obtained from out-of-sample data presumed to be from the same
population. This idea is inspired by the use of genetic reference panels in FVE estimation in genetics \citep{schizophrenia_working_group_of_the_psychiatric_genomics_consortium_ld_2015,taliun_laser_2017}. Our findings indicate that using surrogate eigenvectors for FVE partitioning and estimation yields nearly unbiased and more robust estimates compared to direct applications of GWASH and LMM-REML. We use a data-driven rule motivated by Mar\v{c}enko--Pastur theory \citep{silverstein_analysis_1995} to choose the number of strong eigenvalues: the bulk right edge is estimated using Spectrode \citep{dobriban_efficient_2015} and the eigenvalues above that edge are counted as strong eigenvalues. We present FVE estimation simulation results from artificial datasets based on both one-block and two-block exchangeable correlation structures. The real data analysis shows that approximately $12$\% of the variability in crystallized intelligence can be explained by surface area data following FVE decomposition and estimation in the ABCD dataset.

\begin{figure} 
\includegraphics[width= \linewidth]{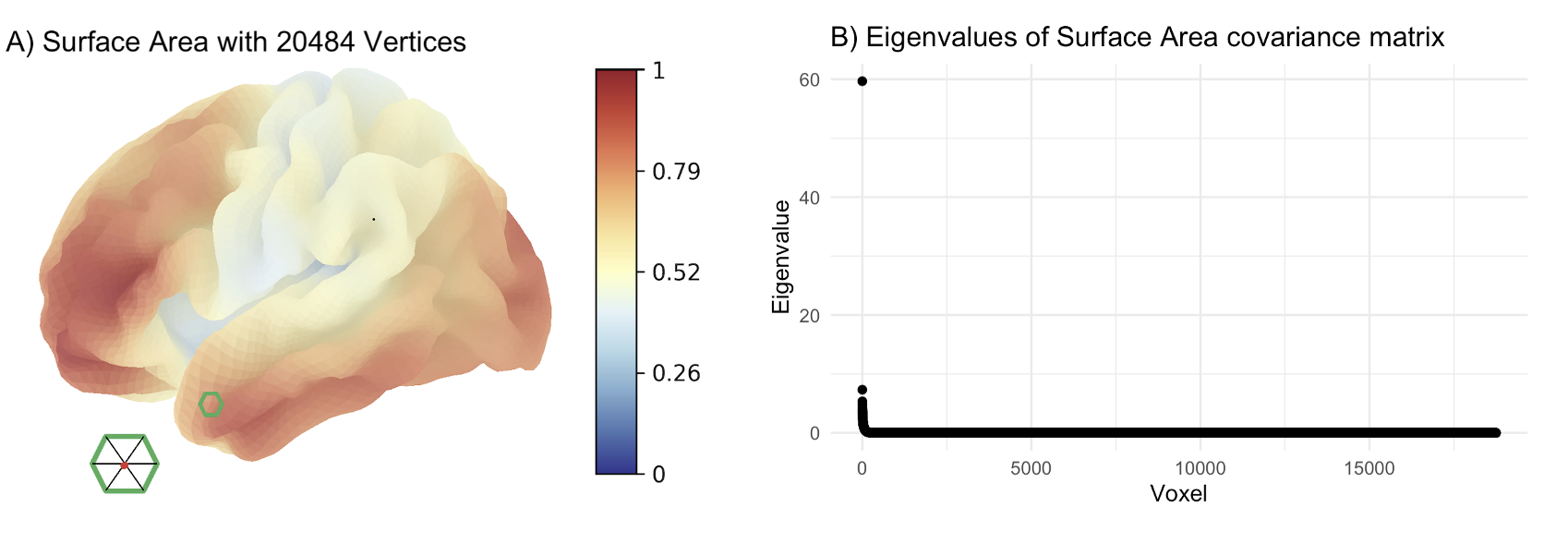} 
\caption{Panel A shows ABCD cortical surface area data in $mm^2$ consisting of 20,484 vertices from one participant at the baseline visit with the left hemisphere plotted. Panel B shows a plot of the eigenvalues of cortical surface area covariance matrix of ABCD data after removing the medial wall.}
\label{Figure1_intro}
\end{figure} 

\section{Definitions and background} 
\subsection{Fraction of Variance Explained}

    We consider the standard linear regression model 
\begin{equation}
\label{eq:1}
    y_i = \vec{\x}_i \boldsymbol{\beta} + \epsilon_i,\qquad i = 1,\ldots,n.
\end{equation}
 Here $y_i$ and $\vec{\x}_i = \left(x_{i1}, ..., x_{im}\right)$ denote the continuous response and  brain imaging predictors of the subject $i$, respectively. In this model, $\boldsymbol{\beta}$ is fixed and arbitrary, with no prespecified distribution. The predictors here represent vertex-wise cortical features, specifically local surface area, but they could represent other voxel-wise brain measurements. Model \eqref{eq:1} can also be written in matrix form as $\boldsymbol{y} = \boldsymbol{X} \boldsymbol{\beta} + \boldsymbol{\epsilon}$, where $\boldsymbol{X}$ is the design matrix with rows $\vec{\x}_i$, $\boldsymbol{y} = \left(y_1, ...,y_n\right)^\top$, $\boldsymbol{\beta} = \left(\beta_{1}, ..., \beta_{m}\right)^\top$, and $\boldsymbol{\epsilon} = \left(\epsilon_1,..., \epsilon_n\right)^\top$.

 For simplicity and without loss of generality, we assume that the vector $\boldsymbol{y}$ and the rows $\boldsymbol{\vec{x}_i}$ have been centered in the sample by subtracting the column-wise average so that $\sum_{i=1}^n {y_i}=0$ and $\sum_{i=1}^n x_{ij}=0$ for all $j=1,\ldots,m$.
 This centering makes model \eqref{eq:1} have no intercept term. Other fixed covariates, such as age, gender, race, and ethnicity, are not explicitly included in our model. Instead, we assume that these variables have already been regressed out in a preliminary regression model. Each participant $i$ is assumed to be randomly sampled from the population together with the associated response $y_i$ and  the brain measurement vertices $\boldsymbol{\vec{x}_i}$. The error terms $\epsilon_i$ are independent of $\boldsymbol{\vec{x}_i}$ with mean 0 and variance $\sigma_{\epsilon}^2$. 
 
 The $m \times m$ covariance matrix $\Sigma = \Cov(\vec{x}_i)$ captures the pairwise covariance structure between the measured values at the cortical vertices in the underlying population. The associated correlation matrix reflects the marginal Pearson correlations between these vertices.  For T1-weighted structural MRI-derived surface area data, this matrix typically exhibits a dominant eigenvector, a property that we leverage in our method (Figure \ref{Figure1_intro}).

We wish to estimate FVE by the predictors in Model \eqref{eq:1}. The variance decomposition is $$\Var\left(y_i|\boldsymbol{\beta}\right) = \E\left(\boldsymbol{\beta}^\top \boldsymbol{\vec{x}_i}^\top \boldsymbol{\vec{x}_i} \boldsymbol{\beta} | \boldsymbol{\beta}\right)  + \E\left(\epsilon_i^2\right)  = \boldsymbol{\beta}^\top \Sigma \boldsymbol{\beta} + \sigma^2.$$ 
Our main quantity of interest is $h^2_{\beta}$, the fixed effect FVE conditional on $\beta$:

\begin{equation}
    \label{eq:h2beta}
    h^2_{\beta} = \frac{\Var\left(\boldsymbol{\vec{x}_i}\boldsymbol{\beta}|\boldsymbol{\beta}\right)}{\Var\left(y_i | \boldsymbol{\beta}\right)} = \frac{\boldsymbol{\beta}^\top \Sigma \boldsymbol{\beta}}{ \boldsymbol{\beta}^\top \Sigma \boldsymbol{\beta} + \sigma^2}.
\end{equation} In comparison, when treating Model \eqref{eq:1} as a random-effects model \citep{azriel_consistency_2025}, the FVE is defined as the random-effects heritability $h^2$:
\begin{equation}
    \label{eq:h2}
    h^2 = \frac{\Var[\boldsymbol{\vec{x}_i}\boldsymbol{\beta}]}{\Var[y_i ]}.
\end{equation} In \cite{schizophrenia_working_group_of_the_psychiatric_genomics_consortium_ld_2015}, $\boldsymbol{\beta}$ is modeled as random: $\beta_1,...,\beta_m$ are assumed i.i.d. with mean zero and independent of $X$ and $\epsilon$. Letting $\Var(y_i) = \sigma^2_y$, $\Var\left(\beta_i\right) = (\sigma^2_yh^2)/(m\Var (x_i))$ one can obtain the random-effects FVE $h^2$ by averaging the numerator and the denominator in Equation~\eqref{eq:h2beta} separately \citep{azriel_consistency_2025}:
\begin{equation*}
    \frac{\E(\boldsymbol{\beta}^\top \Sigma \boldsymbol{\beta})}{ E(\boldsymbol{\beta}^\top \Sigma \boldsymbol{\beta}) + \Var(\epsilon_i)} = \frac{\sigma^2_yh^2}{\sigma^2_yh^2 + \sigma^2_y(1 - h^2)} = h^2.
\end{equation*}

As $\boldsymbol{\beta}$ reflects the underlying unknown biological mechanisms in a given population described by model \eqref{eq:1}, we assume that the parameter $\boldsymbol{\beta}$ is fixed and arbitrary with no prespecified distribution. 
We use the random-effects FVE $h^2$~\eqref{eq:h2} to set the residual variance in the simulation. Because $\boldsymbol{\beta}$ is held fixed across replicates, the relevant true FVE is the fixed-effect FVE $h^2_{\beta}$ against which we evaluate performance.

\subsection{The essential conditions}

Theorem 1 of \citet{azriel_consistency_2025} shows that the fixed-effect FVE \(h_\beta^2\) in \eqref{eq:h2beta} and the random-effects FVE \(h^2\) in \eqref{eq:h2} are asymptotically equivalent under two conditions. The first condition is that the coefficients are of comparable order of magnitude. Since this is not the focus of this paper, we assume that this condition holds as it does for the random-effects model described above. In this paper, we focus on the second condition:
\begin{itemize}
    \item \(\mathrm{WD}_0\) (weak dependence): \(\operatorname{tr}(\widetilde{\Sigma}^2)/m^2 \to 0\) as \(m \to \infty\), \quad where $\widetilde{\Sigma}$ is standardized $\Sigma$.
\end{itemize}
 
As examples, ARMA processes satisfy weak dependence, whereas exchangeable covariance structures violate it. More details can be found in Section 2 of \cite{azriel_empirical_2015}. When the covariance condition $\mathrm{WD}_0$ is not fulfilled, the estimator of $h_{\beta}^2$ based on the sample $y_i's$ and the predictors $\boldsymbol{\vec{x}_i}'s$ becomes biased, as we will show in Section~\ref{section_simulation_result}, in which case $h_{\beta}^2$ and $h^2$ are not asymptotically equivalent. Therefore in the following discussion, we use the fixed-effect FVE $h^2_{\beta}$ as the true FVE when we evaluate simulation performance. 

High-dimensional FVE estimators typically do not perform well under strong correlation \citep{pham_when_2026}. To obtain an unbiased estimator in the presence of strong 
correlations, we decompose the design vector $\vec{x}$ into a low-rank (strong-correlation) component and a high-dimensional (weak-correlation) component, as described in Section~\ref{sec:predictor_decomposition} below.

\subsection{Adjusted $R^2$}

In our proposed FVE decomposition framework, the first portion of the FVE will be estimated in a low-dimensional linear model. This can be estimated using adjusted $R^2$. In classical regression settings with relatively few predictors, the coefficient of determination $R^2$ is a well-established metric for evaluating how much variance in the outcome is explained by the model. However, $R^2$ can be overly optimistic as more predictors are added even if those predictors contribute little real explanatory power. 

The adjusted $R^2$ addresses this shortcoming by introducing a penalty for each additional predictor, thereby preventing spurious improvements in model fit. One of the most widely known and implemented estimators adjusted $R^2$ was developed by \cite{ezekiel1930methods}
\begin{equation}
    R_{adj}^2 
= 1 - \frac{\left(1 - R^2\right)\,\left(n - 1\right)}{n - m - 1}, \label{eq_adjr}
\end{equation}
which requires \(m < n-1\). The Ezekiel adjusted \(R^2\) was shown to be unbiased only if the population multiple correlation is zero \citep{alf_new_2002}.

Wishart (\cite{wishart1931mean}) derived his formula for the variance of $R^2$ using a hypergeometric series, a method that has also been employed in subsequent literature, including the work of Olkin and Finn (\cite{olkin_unbiased_1958}; \cite{stuart2010kendall}; 
\cite{momin_significance_2023}): the estimator of the variance of $R^2$ is given by the following formula
\[
    \widehat{\Var}\left(R^2\right) = \frac{4R^2 \left(1 - R^2\right)^2 \{n - \left(m + 1\right)\}^2}{\left(n^2 - 1\right)\left(n+3\right)}.
\]
Since the adjusted $R^2$ can be written as a linear transformation of $R^2$ in \eqref{eq_adjr}, we can obtain the estimator of the variance of adjusted $R^2$ as
\begin{equation}
\label{adj_var}
     \widehat{\Var}\left(R_{adj}^2\right) = \frac{4R^2 \left(1 - R^2\right)^2 \left(n - 1\right)^2}{\left(n^2 - 1\right)\left(n+3\right)}.
\end{equation}

\subsection{LMM-REML}

Linear mixed-effects models (LMM) with parameters estimated using restricted maximum likelihood (REML) are commonly used to estimate FVE both in genetics, often referred to there as Genomic Restricted Maximum Likelihood (GREML) \citep{yang_common_2010}, or in neuroimaging analysis \citep{sabuncu_morphometricity_2016, couvyduchesne_unified_2020,furtjes_quantified_2023}. Implemented in the software package GCTA (Genome-wide Complex Trait Analysis), it is a powerful statistical method to estimate the FVE by genetic or other high-dimensional data (\cite{yang_gcta_2011}). Unlike traditional regression models that estimate individual predictor effects, LMM-REML leverages a mixed-effects model framework and treats the predictor effects (e.g., SNPs, voxels) as random effects, with a variance-covariance structure defined by a relatedness matrix (e.g., a genetic relatedness matrix, GRM). This matrix captures the pairwise similarity between individuals based on the predictors. 

The core assumption of LMM-REML is that the effects of the predictors are drawn from a multivariate normal distribution with a mean of zero and a covariance structure proportional to the relatedness matrix. Furthermore, LMM-REML assumes that the relatedness matrix accurately reflects the true underlying relationships among individuals and the influence of those factors on the trait, meaning that the matrix design is adequate. Departures from these assumptions, such as strong non-additive genetic effects, unmodeled population structure, or the presence of rare variants not well-captured by the relatedness matrix, can lead to biased estimates of FVE (\cite{krishna_kumar_limitations_2016, the_ucleb_consortium_reevaluation_2017}).

\subsection{GWASH}

An alternative estimator for FVE in high-dimensional settings is GWASH \citep{schwartzman_simple_2019}, which uses summary statistics. To write down this estimator, we first replace $\boldsymbol{X}$ and $\boldsymbol{y}$ by their column-wise standardized versions $\widetilde{\boldsymbol{X}}$ and $\widetilde{\boldsymbol{y}}$. The sample correlation scores are defined as
\[u_j = \frac{\widetilde{x}^\top_j \widetilde{y}}{\sqrt{n -1}} = \sqrt{n -1} \frac{x_j^\top y}{ ||x_j|| ||y||},\qquad j = 1,\ldots,m,\]
where $\widetilde{x}_j$ and $x_j$ are the $j$-th column of  $\widetilde{\boldsymbol{X}}$ and  ${\boldsymbol{X}}$, respectively.
The GWASH estimator is given by
\begin{equation*}
    \hat{h}^2_{ \text{GWASH}} = \frac{m}{n \hat{\mu}_2} \left(s^2 - 1\right),
\end{equation*}
where $s^2$ is the empirical second moment of the correlation scores $s^2 = ||\boldsymbol{u}||^2/m = \sum_{j = 1} ^m u^2_j/m$. The second spectral moment is:  
\begin{equation}
\label{eq:mu2}
    \mu_2 = \tr\left(\widetilde{\Sigma}^2\right)/m.
\end{equation} 
 And $\hat{\mu}_2$ is the estimation of $\mu_2$:
 \begin{equation}
\label{eq:mu2_hat}
    \hat{\mu}_2 = \frac{1}{m} \tr \left( \widetilde{S}^2 \right) - \frac{m-1}{n-1},
\end{equation} 
where $\widetilde{S}$ is the sample covariance matrix of $\widetilde{\Sigma} = \Cor(\vec{\bf x}_i)$ and is given by
\begin{equation*}
\label{eq:S_stdX}
    \widetilde{S} = \frac{1}{n - 1} \widetilde{\boldsymbol{X}}^T\widetilde{\boldsymbol{X}}.
\end{equation*}

Under certain assumptions,
the asymptotic distribution of $\hat{h}^2_{\text{GWASH}}$ is normal when $m$ and $n$ increase so that the ratio between $m$ and $n$ converges to a constant, \[\sqrt{n} \left(\hat{h}^2_{ \text{GWASH}} -h^2\right)/\psi \quad {\to}\quad N\left(0, 1\right),\]where $$\frac{\psi^2}{n} = \frac{2}{n}\left(\frac{m}{n \mu_2}+ 2\frac{\mu_3}{\mu_2^2}h^2 - h^4\right),$$ and $\mu_3 =  \tr\left(\widetilde{\Sigma}^3\right)/m$ \citep{schwartzman_simple_2019}. An estimate of the standard error of $h^2_{\text{GWASH}}$ is $\sqrt{\hat{\psi}^2/n}$, where 
\begin{equation}
\label{gwash_var}
    {\hat{\psi}^2}  = 2\left(\frac{m}{n \hat{\mu}_2}+ 2\frac{\hat{\mu}_3}{\hat{\mu}_2^2}\hat{h}^2_{\text{GWASH}} - \hat{h}_{\text{GWASH}}^4\right),
\end{equation} and $\hat{\mu}_2$ is given in \eqref{eq:mu2_hat} and \[\hat{\mu}_3 = \frac{1}{m} \tr\left( \widetilde{S}^3 \right) - 3\frac{m-1}{n-1} \hat{\mu}_2 - \frac{(m-1)(m-2)}{\left( n-1\right)^2}. \] 

\section{Methods}

We first describe the decomposition of the predictors and the correlation structures that motivate removing leading PCs. We then use this predictor decomposition to rewrite the response model and derive the corresponding FVE decomposition and estimation procedure.

\subsection{Predictor decomposition}
\label{sec:predictor_decomposition}
Let $\Sigma = V \Lambda V^\top$ be the eigen-decomposition of the covariance 
matrix, and partition the eigenstructure as
\begin{equation*}
\label{sigma_decomposition}
    \Sigma = \Sigma_{k} + \Sigma_{\cdot k} = V_k \Lambda_k V_k^\top + V_{.k} \Lambda_{.k} V_{.k}^\top,
\end{equation*}
where $\Lambda = \operatorname{diag}(\lambda_{1},\ldots,\lambda_{m}), \quad \lambda_{1}\ge\cdots\ge\lambda_{m}\ge 0
$ are the eigenvalues and $V_k$ collects the first $k$ eigenvectors corresponding to the largest 
eigenvalues, with $V_{.k}$ containing the remaining $m-k$ eigenvectors. 

We define the orthogonal projections
\begin{equation}
\label{eq:predictor_decomposition}
\vec{\x}_k = \vec{\x} V_k V_k^\top, \qquad
\vec{\x}_{.k} = \vec{\x} V_{.k} V_{.k}^\top,\qquad
\vec{\x} = \vec{\x}_k + \vec{\x}_{.k}.
\end{equation}
The corresponding transformed predictor coordinates are
\[
\vec{\w}_k = \vec{\x} V_k, 
\qquad  
\vec{\w}_{.k} = \vec{\x} V_{.k}.
\]
Since $V_k^\top V_{.k} = 0$, the covariance matrix of 
$(\vec{\w}_k,\vec{\w}_{.k})$ is diagonal:
\begin{equation}
\label{eq:block-independent}
    \E
\begin{bmatrix}
\vec{\w}_k^\top \\[2pt] 
\vec{\w}_{.k}^\top
\end{bmatrix}
\begin{bmatrix}
\vec{\w}_k & \vec{\w}_{.k}
\end{bmatrix}
=
\begin{bmatrix}
\Lambda_k & 0 \\
0 & \Lambda_{.k}
\end{bmatrix},
\end{equation}
the two transformed predictor blocks are uncorrelated. This separates the predictor space into a leading $k$-dimensional component and its orthogonal complement.

\subsection{The number of strong eigenvalues $K$}

To determine the number of ``strong'' eigenvalues that are of order $m$, we follow the work of \citet{azriel_empirical_2015}.
Define the number of strong eigenvalues by
\[
K \;=\; \sum_{i\ge 1} \mathbf{1}\!\left\{\lim_{m\to\infty}\frac{\lambda_i}{m}>0\right\},
\]
and assume $K<\infty$. We call the remaining eigenvalues the bulk. Throughout, we also assume the bulk is uniformly bounded, i.e.,
\begin{equation}
\label{eq:bulk_bound}
\sup_{m}\,\max_{i>K}\lambda_i \leq C < \infty.
\end{equation}

We first establish that removing a finite sufficient number of population-level
``strong'' PCs eliminates strong global dependence and
restores weak correlation in the residual covariance. Intuitively, when only finitely many eigenvalues grow proportionally with dimension, projecting the data onto the orthogonal
complement of the corresponding eigenvectors removes the dominant
low-rank structure, leaving a residual covariance that satisfies the weak-correlation condition.

Let $V_k\in\mathbb{R}^{m\times k}$ collect the top $k$ orthonormal eigenvectors of $\Sigma$, and define the row-residual by right projection
\[
\vec{\x}_{\cdot k}\;=\; \vec{\x}\,(I - V_ kV_ k^\top)\in \mathbb{R}^{1\times m}.
\]
Then the residual covariance equals
\begin{equation}
\label{eq:residual_covariance_matrix}
    \Sigma_{\cdot k}\;:=\;\operatorname{Cov}(\vec{\x}_{\cdot k}^\top)
\;=\;(I - V_ k V_ k^\top)\,\Sigma\,(I - V_ k V_ k^\top).
\end{equation}
Define $D_{\cdot k}:=\diag(\Sigma_{\cdot k})$
and the re-standardized residual correlation matrix
\begin{equation}
\label{eq:residual_correlation_matrix}
    \tilde{\Sigma}_{\cdot k}
    :=
    D_{\cdot k}^{-1/2}
    \Sigma_{\cdot k}
    D_{\cdot k}^{-1/2}
    =
    \Cor(\vec{\x}_{\cdot k}).
\end{equation}
Assume that the residual variances are uniformly bounded away from zero:
\begin{equation}
\label{eq:residual_variances_bounded}
\inf_m \min_{1\le j\le m}
    (\Sigma_{\cdot k})_{jj}
    \ge c_{k}>0.
\end{equation}

\begin{proposition}
[Population PC removal with at least $k \ge K$ strong factors]\label{prop:pop-AS} 
The residual correlation matrix \eqref{eq:residual_correlation_matrix} satisfies the weak-correlation condition
\begin{equation}
\label{eq:weak_cor}
    \frac{1}{m^2}\,\tr\!\big(\widetilde{\Sigma}_{\cdot k}^{\,2}\big)\ \longrightarrow\ 0\qquad (m\to\infty).
\end{equation}
\end{proposition}

In practice, the population eigenvectors are unknown, and the number of strong
components must be inferred from the data. We adopt a practical data-driven method, motivated by Mar\v{c}enko--Pastur (MP) theory, to choose $\hat{k}$ in Section~\ref{sec:edge_algorithm}.

\subsection{Examples of strong correlation structures}
\label{section_example_derivation}
The following examples illustrate how strong components arise in simple correlation structures and how the residual covariance behaves after the leading components are removed. Details of the derivation can be found in Appendix~\ref{appendix:derivation_for_example}.

\paragraph{One-block exchangeable correlation.}

Let $\vec{x}$ be exchangeable with entries satisfying $$ \Cor(x_i, x_j) = \rho \quad (i \neq j), \quad \rho \geq 0,$$ so that the correlation matrix is 
\begin{equation}
\label{eqone-exchangable}
    \Sigma = \begin{bmatrix} 1 & \rho & \cdots & \rho \\ \rho & 1 & \cdots & \rho \\ \vdots & \vdots & \ddots & \vdots \\ \rho & \rho &\cdots & 1 \end{bmatrix}.
\end{equation}
The eigenvalues of $\Sigma$ are $$\lambda_{1} = 1 + (m - 1)\rho \quad \text{and} \quad \lambda_{2} = \dots = \lambda_{m} = 1 - \rho,$$ with leading eigenvector proportional to $\mathbf{1}_m$ as $(1,\dots,1)^\top /\sqrt{m}$. After standardization and for fixed $\rho>0$, the weak-dependence condition $\mathrm{WD}_0$ fails since
\begin{equation}
\label{eq:one_block_nopcremove_strong}
    \frac{1}{m^2}\tr(\widetilde{\Sigma}^2)
= \frac{1}{m^2} \sum_{i = 1}^m \lambda_{i}^2
\longrightarrow \rho^2>0.
\end{equation}

However, removing the population leading PC following \eqref{eq:residual_covariance_matrix} in Proposition~\ref{prop:pop-AS}, where $K=1$ and $V_1V_1^\top=\frac{1}{m}\mathbf 1_m\mathbf 1_m^\top$,
yields the residual covariance
\[
\Sigma_{\cdot 1}=(1-\rho)(I_m-V_1V_1^\top),
\]
because $(I-V_1V_1^\top)\mathbf 1_m=0$. After removing the first PC, the re-standardized residual correlation matrix in Equation~\eqref{eq:residual_correlation_matrix} is:
\[
    \widetilde{\Sigma}_{\cdot 1}
    :=
    D_{\cdot 1}^{-1/2}
    \Sigma_{\cdot 1}
    D_{\cdot 1}^{-1/2}
    =
    \frac{m}{m-1}
    \left(
    I_m-\frac{1}{m}\mathbf 1_m\mathbf 1_m^\top
    \right).
\] Therefore $\widetilde{\Sigma}_{\cdot 1}$ has eigenvalues $0$ and
$m/(m-1)$, the latter with multiplicity $m-1$. Hence the second spectral
moment \eqref{eq:mu2} converges to $1$ as $m \to \infty$:
\begin{equation*}
\label{eq_second_spectral_moment_oneblock}
    \mu_2 =\frac{1}{m}\tr(\widetilde{\Sigma}_{\cdot 1}^2)=\frac{m}{m-1}\to1.  
\end{equation*}
Therefore, following the definition in \eqref{eq:weak_cor},
\[
\frac{1}{m^2}\tr(\widetilde{\Sigma}_{\cdot 1}^2)=\frac{1}{m-1}\longrightarrow 0,
\]
verifying Proposition~\ref{prop:pop-AS} in this model. 

\paragraph{Two-block exchangeable correlation.} The number of leading PCs with strong eigenvalues can exceed one. Consider a two-block structure of the form \eqref{eqone-exchangable} with block sizes $m_1+m_2=m$, within-block
correlations $\rho_1,\rho_2$, and between-block correlation $\rho_B$. 
We write the full correlation matrix as
\begin{equation}
\label{eqtwo-exchangable}
\Sigma = 
\begin{bmatrix}
(1-\rho_1)I_{m_1}+\rho_1\, \mathbf{1}_{m_1}\mathbf{1}_{m_1}^\top & \rho_B\, \mathbf{1}_{m_1}\mathbf{1}_{m_2}^\top\\[4pt]
\rho_B\, \mathbf{1}_{m_2}\mathbf{1}_{m_1}^\top &
(1-\rho_2)I_{m_2}+\rho_2\, \mathbf{1}_{m_2}\mathbf{1}_{m_2}^\top
\end{bmatrix}.
\end{equation}
Following the two-block eigen-decomposition in \citet{azriel_empirical_2015}, define
\[
u_1=
\begin{bmatrix}\mathbf{1}_{m_1}/\sqrt{m_1}\\ \mathbf{0}_{m_2}\end{bmatrix},
\qquad
u_2=
\begin{bmatrix}\mathbf{0}_{m_1}\\ \mathbf{1}_{m_2}/\sqrt{m_2}\end{bmatrix},
\]
and let
\begin{equation}
\label{eq:eigenspace_two_block}
    S=\operatorname{span}\{u_1,u_2\}\subset\mathbb{R}^{m_1+m_2}.
\end{equation}

In the orthonormal basis $\{u_1,u_2\}$, the restriction of $\Sigma$ to $S$ is represented by
\begin{equation*}
    M_m:=
\begin{bmatrix}
a_1 & \sqrt{m_1m_2}\,\rho_B\\
\sqrt{m_1m_2}\,\rho_B & a_2
\end{bmatrix},
\end{equation*}
where
\[
a_1=1+(m_1-1)\rho_1,
\qquad
a_2=1+(m_2-1)\rho_2.
\]
Thus the two strong eigenvalues are
\[
\lambda_{\pm}(M_m)
=
\frac{a_1+a_2}{2}
\pm
\left\{
\left(\frac{a_1-a_2}{2}\right)^2
+
m_1m_2\rho_B^2
\right\}^{1/2}.
\]The remaining eigenvalues of $\Sigma$ are \(1-\rho_1\) with multiplicity \(m_1-1\) and \(1-\rho_2\) with multiplicity \(m_2-1\).

We consider the following asymptotic cases. Suppose first that $\rho_1$ and $\rho_2$ are fixed positive constants, and that $m_j/m\to\pi_j$ for $j=1,2$. After standardization when no PC is removed, the weak-dependence condition $\mathrm{WD}_0$ fails:
\begin{equation}
\label{eq:two_block_nopcremove_strong}
\frac{1}{m^2}\tr(\widetilde{\Sigma}^2)
\longrightarrow
\pi_1^2\rho_1^2
+
\pi_2^2\rho_2^2
+
2\pi_1\pi_2\rho_B^2.
\end{equation}If $\pi_1,\pi_2>0$ and $\rho_B^2<\rho_1\rho_2$, then both eigenvalues $\lambda_+(M_m)$ and $\lambda_-(M_m)$ are of order $m$. The standardized residual correlation matrix in Equation~\eqref{eq:residual_correlation_matrix} after removing the top
two population PCs is
\begin{align*}
\widetilde{\Sigma}_{\cdot 2}
&=
D_{\cdot 2}^{-1/2}
\Sigma_{\cdot 2}
D_{\cdot 2}^{-1/2} \notag\\
&=
\begin{bmatrix}
\frac{m_1}{m_1-1}
\left(
I_{m_1}
-
\frac{1}{m_1}
\mathbf 1_{m_1}\mathbf 1_{m_1}^\top
\right)
& 0\\[6pt]
0 &
\frac{m_2}{m_2-1}
\left(
I_{m_2}
-
\frac{1}{m_2}
\mathbf 1_{m_2}\mathbf 1_{m_2}^\top
\right)
\end{bmatrix}.
\end{align*}Removing the top $K=2$ population PCs eliminates the entire $O(m)$ part of the spectrum and leaves only bounded eigenvalues $\{1-\rho_1,1-\rho_2\}$ and hence according to Proposition~\ref{prop:pop-AS}:
\[\frac{1}{m^2}
    \tr\!\left(\widetilde{\Sigma}_{\cdot 2}^{\,2}\right)
    =
    \frac{
    \frac{m_1^2}{m_1-1}
    +
    \frac{m_2^2}{m_2-1}
    }{m^2} = \frac{m + O(1)}{m^2} = O\left(\frac{1}{m}\right)
    \longrightarrow 0.
\]If $\pi_1,\pi_2>0$ but $\rho_B^2=\rho_1\rho_2$, then the leading $O(m)$ part of $M_m$ has rank one. Consequently, only $\lambda_+(M_m)$ is of order $m$, while $\lambda_-(M_m)$ is bounded, and hence there is only one strong eigenvalue and $K = 1$. If one block has size $o(m)$ while the other block has size of order $m$ with positive within-block correlation, then again only one eigenvalue is of order $m$ and $K = 1$. 
\subsection{Estimating the number of large eigenvalues $K$}
\label{sec:edge_algorithm}

To estimate the number of large eigenvalues $K$ in real data, we use the Mar\v{c}enko--Pastur law in random matrix theory \citep{marcenko_distribution_1967}. In their nomenclature, the word `bulk' refers to the main cluster of bounded population eigenvalues of $\Sigma$, and
the word `spikes' refers to a small number of larger separated eigenvalues. If the empirical
distribution of the bulk eigenvalues converges to a law $H$, then the sample bulk
spectrum is described by the Mar\v{c}enko--Pastur law $F_{\gamma,H}$ \citep{silverstein_analysis_1995,bai_spectral_2010}. 

Our goal is to construct a working cutoff for the empirical bulk, so that any eigenvalue above this cutoff is flagged as a candidate strong component. We use Spectrode \citep{dobriban_efficient_2015} to carry out this calculation numerically. Spectrode computes the support and density of the limiting spectral distribution; in this paper, we use only its support calculation and extract the rightmost bulk endpoint and denote this
edge by $\hat\nu$. We then define our data-driven choice for the number of strong eigenvectors $\hat{K}$ as:
\[
\hat K
=
\#\bigl\{j:\hat\lambda_j>\hat\nu\bigr\}.
\]

In Figure~\ref{fig:sim_twoblock_result}, we show the comparison of unit MP spectral distribution (solid green line) and the distribution of the eigenvalues after rescaling by the estimated right edge $\hat\nu$ from Spectrode. Exactly one and two eigenvalues are separated from the bulk and exceed $\hat\nu$, leading to $\hat K=1$ and $\hat K=2$ under the one-block and two-block exchangeable correlation structure respectively. In these examples, the procedure correctly distinguishes the strong eigenvalue components from the empirical bulk spectrum.

\subsection{FVE decomposition}
\label{sec:FVE_decomposition_theorem}
We now return to the linear model \eqref{eq:1} to describe the FVE decomposition. Define the coefficient coordinates in the coordinate frame of the first k eigenvectors:
\[
\bgamma_k = V_k^\top \bbeta,
\qquad
\bgamma_{.k} = V_{.k}^\top \bbeta.
\]
Substituting the predictor decomposition in Equation \eqref{eq:predictor_decomposition} into model \eqref{eq:1} yields two equivalent representations:
\begin{equation*}
y = \vec{\x}_k \bbeta + \vec{\x}_{.k} \bbeta + \epsilon,
\end{equation*}
\begin{equation}
y = \vec{\w}_k \bgamma_k + \vec{\w}_{.k} \bgamma_{.k} + \epsilon. \label{eqform2}
\end{equation}
Since the two transformed predictor blocks are uncorrelated as shown in~\eqref{eq:block-independent}, we can rewrite the model \eqref{eq:1} as a two-stage model with the projected representations \eqref{eqform2}:
\begin{enumerate}
    \item \textbf{First stage (strong effects):}
    \[
        y = \vec{\x}_k \bbeta+ \eta = \vec{\w}_k \bgamma_k + \eta, \quad \text{where } \eta = \vec{\w}_{.k} \bgamma_{.k} + \epsilon.
    \]
    \item \textbf{Second stage (weak effects):}
    \[
        y_{.k} = \eta = y - \vec{\w}_k \bgamma_k = \vec{\w}_{.k} \bgamma_{.k} + \epsilon.
    \]
\end{enumerate}We define the leading-PC contribution to the total FVE as
\begin{equation}
\label{eq:hk_def}
h_k^2 \;:=\; \frac{\Var(\vec{\w}_k\bgamma_k)}{\Var(y)},
\end{equation}
and the residual FVE for the second-stage model as
\begin{equation}
\label{eq:hdotk_def}
h_{.k}^2 \;:=\; \frac{\Var(\vec{\w}_{.k}\bgamma_{.k})}{\Var(y_{.k})}.
\end{equation}
According to these definitions, $h_k^2$ quantifies the FVE in $y$ explained by the top-$k$ PC subspace, while $h_{.k}^2$ quantifies the FVE in the residualized outcome $y_{.k}$ explained by the remaining $m-k$ directions. The following proposition shows how the total FVE $h^2_\beta$, defined in Equation~\eqref{eq:h2beta} can be written as a combination of these two partial FVE components.

\begin{proposition}[FVE decomposition]\label{prop:pc_fve_decomp}
Under \eqref{eq:1} and the definitions above, the total fixed-effect FVE
$h_\beta^2$ in \eqref{eq:h2beta} decomposes as
\begin{equation}\label{eq:pc_fve_identity}
h_\beta^2 \;=\; h_k^2 + (1-h_k^2)\,h_{.k}^2.
\end{equation}
\end{proposition}

\subsection{FVE estimation}
\label{sec:FVE_estimation_procesure}

The proposed FVE decomposition provides a principled way to address strong correlation in high-dimensional settings. By projecting and residualizing both the design matrix and the response onto the orthogonal complement 
of the leading eigenvectors, the decomposition separates the total FVE into two complementary terms, see Proposition~\eqref{prop:pc_fve_decomp} and Equation~\eqref{eq:pc_fve_identity}. The first term corresponds to a low-dimensional domain, while the second term accounts for the residual variance in high-dimensional domain. 

In Section~\ref{sec:predictor_decomposition} and Section~\ref{sec:FVE_decomposition_theorem}, the matrix of leading eigenvectors $V_k$ comes from the population covariance matrix $\Sigma$. In practice, the population eigenvectors are unknown and the eigenvectors $\widehat V_k$ of the study-sample covariance matrix are a natural substitute. However, when the same dataset $\boldsymbol{X}$ is used both to estimate the eigenvectors and to construct the residualized predictors
\[
\boldsymbol{X}_{.k} = \boldsymbol{X}(I - \widehat V_k \widehat V_k^\top),
\]
the procedure suffers from the double-dipping problem \citep{kriegeskorte_circular_2009, ball_double_2020}. The estimated eigenvectors 
$\widehat V_k$ absorb sample-specific noise, so the projection 
$I - \widehat V_k \widehat V_k^\top$ removes not only true population structure 
but also idiosyncratic noise directions, artificially weakening the residual 
correlation structure and biasing the high-dimensional portion of the FVE estimate, as we will show in the simulations in Section~\ref{sec:simulaiton} below.

To avoid this dependence, we use surrogate eigenvectors $\widehat{V}_{\mathrm{sur},k}$ estimated from an independent reference dataset whose covariance structure closely matches that of the study sample. The purpose of the surrogate eigenvectors is to break the feedback loop created by double dipping: conditional on the reference dataset, the eigenvectors are fixed before the FVE estimator is applied to the study sample. Similar reference-panel strategies are common in genetics; for example, ancestry methods estimate a PCA reference space from genotyped reference individuals and then project study samples into that space \citep{schizophrenia_working_group_of_the_psychiatric_genomics_consortium_ld_2015, taliun_laser_2017}. We compare population, sample, and surrogate eigenvectors via simulations in Section~\ref{sec:simulaiton}. For the analyses presented subsequently, we assume that the surrogate dataset is of comparable size to the primary dataset. Limited surrogate sample size may introduce noise, and its impact is examined in Appendix~\ref{surrogate_sample_datasplitting}. 

Using surrogate eigenvectors, the predictor residualization according to Proposition~\ref{prop:pop-AS} is defined as: 
\begin{equation}
\label{eq:right_equation}
    \boldsymbol{X}_{.k} = \boldsymbol{X}(I - \widehat{V}_{\mathrm{sur},k}\widehat{V}_{\mathrm{sur},k}^\top),
\qquad
W_k = \boldsymbol{X}\widehat{V}_{\mathrm{sur},k}.
\end{equation}
Residualizing the response vector must be performed in sample space, which leads to the following linear residuals:
\begin{equation}
        \boldsymbol{y}_{.k}  = \boldsymbol{y} - W_k(W_k^\top W_k)^{-1} W_k^\top \boldsymbol{y}.
        \label{left_equation_linear} 
\end{equation} 
When $k$ is small relative to the sample size $n$, $(W_k^\top W_k)^{-1}$ can be computed stably. For large $k$, however, the matrix $W_k^\top W_k$ can become ill-conditioned, and inverting it may amplify finite-sample estimation noise. We therefore also consider the diagonal approximation:\begin{equation}
     \boldsymbol{y}^{(\mathrm{diag})}_{.k}
     = \boldsymbol{y} - W_k(D_{k}^{2})^{-1}W_k^\top \boldsymbol{y},
     \label{left_equation_diagonal}
\end{equation}
where $D_{k}$ contains the first $k$ singular values of $\boldsymbol{X}$.  

We propose the following procedure for estimating the decomposed FVE:
\begin{enumerate}
    \item Estimate $\hat K$ according to Section~\ref{sec:edge_algorithm}. 
    \item Choose $k \ge \hat K$ and obtain the corresponding surrogate eigenvectors $\widehat{V}_{\mathrm{sur},k}$.

    \item Form the matrix $W_k = \boldsymbol{X}\widehat{V}_{\mathrm{sur},k}$ and residualize the predictors using projection \eqref{eq:right_equation}.

    \item Residualize the response using the linear projection \eqref{left_equation_linear} or linear-diagonal projection \eqref{left_equation_diagonal}.

    \item Estimate the low-dimensional FVE $h^2_{k}$~\eqref{eq:hk_def} from $(\boldsymbol{y},W_k)$ with Ezekiel adjusted $R^2$ \eqref{eq_adjr}.
    \item Estimate the high-dimensional FVE $h^2_{.k}$~\eqref{eq:hdotk_def} from $(\boldsymbol{y}_{.k},\boldsymbol{X}_{.k})$ using GWASH or LMM-REML.

    \item Combine the two components as a plug-in estimator from Equation~\eqref{eq:pc_fve_identity} in  Proposition~\ref{prop:pc_fve_decomp}:
    \[
        \hat{h}^2_{\beta}
        = \hat{h}^2_{k} + (1-\hat{h}^2_k)\, \hat{h}^2_{.k}.
    \]
\end{enumerate}

The variance of the estimators 
$\hat{h}^2_{k}$ and $\hat{h}^2_{.k}$ can be estimated through \eqref{adj_var} and \eqref{gwash_var} for GWASH or GCTA software for LMM-REML, respectively. However, the $\hat{h}^2_{k}$ and $\hat{h}^2_{.k}$ are not independent as they are computed from the same study outcome $y$. For simplicity, we propose an estimator of variance where we ignore the covariance:
\begin{equation}
\label{eq:var_formula_FVE_decom}
    \widehat{\Var}(\hat{h}^2_{\beta})
    = (1-\hat{h}^2_{.k})^2 \widehat{\Var}(\hat{h}^2_k)
    + (1-\hat{h}^2_k)^2 \widehat{\Var}(\hat{h}^2_{.k}).
\end{equation} The detailed derivation is presented in Appendix~\ref{appendix:delta_var_decomp}. We show via simulations that this estimator works well, details are presented in Appendix~\ref{appendix:SE_comparison}.
\section{Simulations}
\label{sec:simulaiton}
The first simulation provides an empirical verification of Proposition~\ref{prop:pop-AS} by examining whether the second spectral moments of the covariance matrix in Equation~\eqref{eq:weak_cor} converge to the theoretical limits before and after removing the leading PCs. Proposition~\ref{prop:pop-AS} shows that, once at least the top $k \ge K$ strong components are removed, the residual covariance matrix should satisfy the weak-correlation condition and converge to their theoretical limits as $m$ and $n$ increase. The second set of simulations compares our proposed FVE decomposition method, where low-dimensional FVE is estimated using adjusted $R$ squared and high-dimensional FVE using GWASH and LMM-REML, with directly estimating FVE using GWASH and LMM-REML. 

\subsection{Spectral moments}
\label{sec:spectral_moment_simulation}
We evaluate the convergence of the second spectral moment under scenarios where zero or $K$ PCs are removed, where the residual covariance matrices are computed using population, sample, or surrogate eigenvectors under one- and two-block exchangeable correlation structure. As shown in Section~\ref{section_example_derivation}, removing the leading population PC from the one-block exchangeable correlation structure leaves a residual covariance structure that satisfies the weak-correlation condition, similarly in the two-block exchangeable correlation structure with $K = 2$. 

In this simulation, we generate $n$ rows of $\boldsymbol{X}$ of size $m$ i.i.d. from the multivariate normal distribution with mean $\vec{ \boldsymbol{0}}$ and covariance matrix $\Sigma$. We consider two covariance structures:
\begin{itemize}
        \item Covariance matrix $\Sigma$ is of one-block exchangeable correlation structure with parameter $\rho$ following Equation~\eqref{eqone-exchangable}.
        \item Covariance matrix $\Sigma$ is of two-block exchangeable correlation structure where $\rho_B= 0$ following Equation \eqref{eqtwo-exchangable}.
\end{itemize}
The number of observations $n$ increases from $500$ to $2000$ while the ratio between $m/n$ stays as $2$. For each configuration, we conduct 1000 simulation replicates and compare the estimated second spectral moments in Equation~\eqref{eq:weak_cor} with the asymptotic spectral moments calculated from Equation~\eqref{eq:one_block_nopcremove_strong} and Equation~\eqref{eq:two_block_nopcremove_strong} when no PC is removed, and with 0 when adequate amount of $K$ PCs are removed.

When no PC is removed from $\boldsymbol{X}$, we compare the performance of the \(\mathrm{WD}_0\) with sample correlation matrix plug-in estimator~\eqref{eq:weak_cor} and bias-corrected estimator based on \eqref{eq:mu2_hat}. In the case where $K$ PCs are removed from $\boldsymbol{X}$, we compare the sample spectral moment estimators under different choices of eigenvectors: population eigenvectors, sample eigenvectors and surrogate eigenvectors.

\begin{figure}
    \centering
\includegraphics[width=\linewidth]{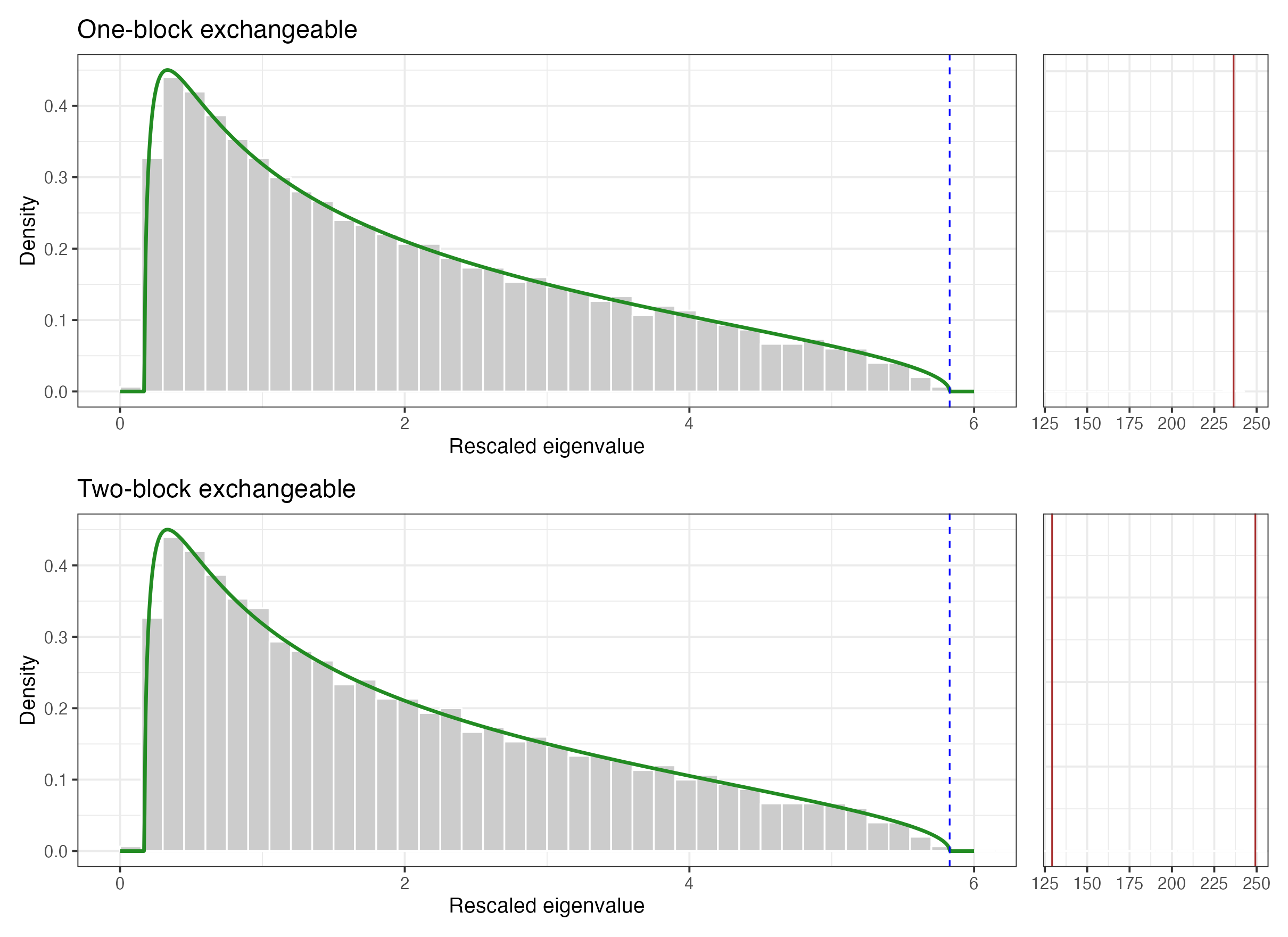}
    \caption{
   The gray histograms show the eigenvalues (brown solid lines) after rescaling by the right edge estimated from Spectrode and the solid green curves show the spectral distribution based on the unit MP law. The first row corresponds to the one-block exchangeable correlation structure, and the second row corresponds to the two-block exchangeable correlation structure. The dotted red lines in the separate right panels indicate the strong eigenvalues for the one-block and two-block exchangeable correlation structures, respectively.
    }
    \label{fig:sim_twoblock_result}
\end{figure}

\begin{figure} 
\includegraphics[width= \linewidth]{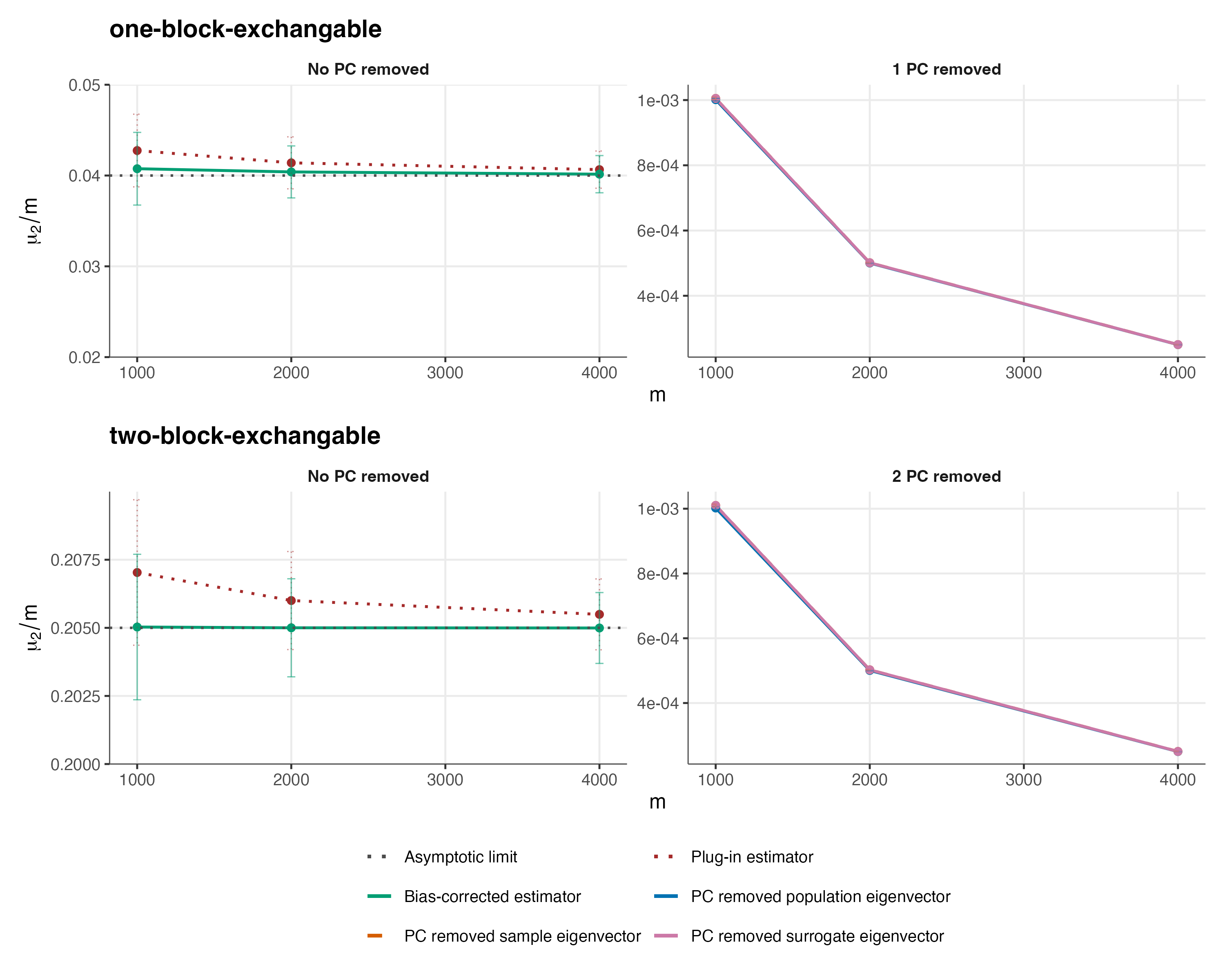} 
\caption{Monte Carlo simulation results of estimators for the second spectral moments as a function of $m$. The first row shows results in a one-block exchangeable correlation structure and the second row shows results under a two-block exchangeable correlation structure. The left column reports the spectral-moments estimates with no PC removed and the right column reports results after removing the leading $K$ PC using population (blue solid line), sample (orange dashed line), or surrogate (pink solid line) eigenvectors. The grey dotted line in the left column is the asymptotic line following Equation~\eqref{eq:one_block_nopcremove_strong} and Equation~\eqref{eq:two_block_nopcremove_strong} for one-block and two-block, respectively. The dotted brown line denotes the plug-in estimator and the solid green line shows the bias-corrected estimator; points are simulation averages with error bars indicating Monte Carlo variability.}
\label{sim_1_result}
\end{figure} 

The simulation results for the spectral moment are presented in Figure~\ref{sim_1_result}. The left column demonstrates that without removing any PC, the correlation structure of the matrix $\boldsymbol{X}$ does not exhibit weak dependence, as the spectral moments converge to the asymptotic line following Equation~\eqref{eq:one_block_nopcremove_strong} and Equation~\eqref{eq:two_block_nopcremove_strong} for one-block and two-block structures, respectively. In contrast, after removing $K$ PCs, as shown in the right column, the spectral moment converges to 0 as the sample size increases. This convergence is observed regardless of whether population, sample, or surrogate eigenvectors are used, indicating that the removal of the leading component effectively eliminates the strong correlations in the covariance structure of $\boldsymbol{X}$ as shown in Proposition~\ref{prop:pop-AS}.

\subsection{Simulation setting of FVE decomposition}
\label{sec_FVE_decomposition}
In this simulation study, we compare several scenarios for estimating the FVE with a decomposition-based approach as in Section \ref{sec:FVE_estimation_procesure}. We first evaluate the impact of double dipping from using sample eigenvectors under one-block exchangeable correlation structure. Subsequently, we assess the decomposition method with surrogate eigenvectors in conjunction with GWASH and LMM-REML, comparing their performance against directly applying GWASH or LMM-REML without decomposition. This allows us to examine whether the decomposition step mitigates the FVE estimation bias and to evaluate the stability of each method when more PCs are removed than necessary. 

The simulation setting is as follows:


\begin{enumerate}

    \item Generate $n$ i.i.d. rows of $\boldsymbol{X}$ from a multivariate normal distribution with mean $\vec{\boldsymbol{0}}$ and covariance matrix $\Sigma$ as in Section \ref{sec:spectral_moment_simulation}.
    \item Generate $\bbeta$ once for each simulation scenario from $N(0,\frac{h^2}{m})$, then hold it fixed across replicates. For each replicate, generate $n$ i.i.d. errors $\epsilon$ from $N(0,1-h^2)$.
    \item Generate $ y_i = \boldsymbol{\vec{x}_i} \boldsymbol{\beta} + \epsilon_i$, $i = 1,\ldots,n$ as in model~\eqref{eq:1} and calculate $h^2_{\beta}$ according to Equation~\eqref{eq:h2beta}.
    \item Conduct the FVE estimation procedure as described in Section \ref{sec:FVE_estimation_procesure} to obtain $\hat{h}^2_{\beta}$.
    \item After 1000 simulations, calculate the mean of $\hat{h}^2_{\beta}$ and empirical standard error to compare with $h^2_{\beta}$. 
\end{enumerate}

\subsection{Simulation results of FVE decomposition}   
\label{section_simulation_result}

\begin{figure} 
\includegraphics[width= \linewidth]{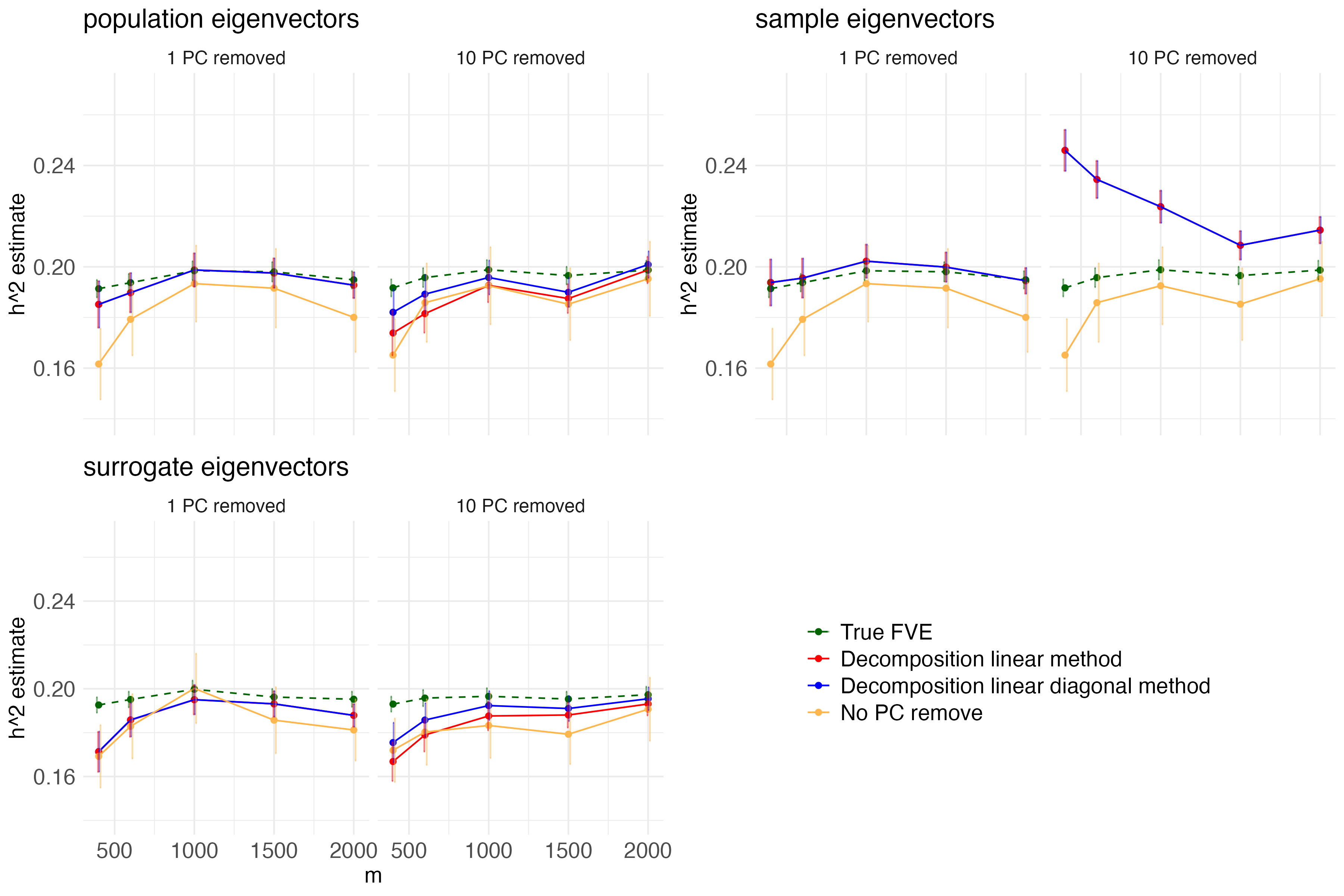} 
\caption{Monte Carlo simulation results for FVE estimation under a one-block exchangeable correlation structure with Adjusted R squared for low-dimensional FVE and GWASH for high-dimensional FVE, comparing population, sample, and surrogate eigenvectors. Within each eigenvector panel, the left and right subplots correspond to removing 1 and 10 PCs, respectively. Colored curves denote the decomposition-based estimators: red solid line for the linear method in Equation~\eqref{left_equation_linear} and blue solid line for the linear diagonal method in Equation~\eqref{left_equation_diagonal}.  FVE estimation with no PC removed is presented in orange solid line and the dashed green line indicates the true FVE. Points are simulation averages with error bars indicating Monte Carlo variability.}
\label{sim_eigenvec_result}
\end{figure}

As described in Section \ref{sec_FVE_decomposition} above, we perform two sets of FVE simulations. The first set of FVE simulations examines which eigenvectors used for decomposition affect the total-FVE estimate under the one-block exchangeable structure (Figure~\ref{sim_eigenvec_result}). The three eigenvector panels compare population, sample, and surrogate eigenvectors, and within each panel the left and right subplots correspond to removing one and ten PCs, respectively. The no-PC-removal estimator (orange line) is consistently below the true FVE, showing the downward bias caused by the strong exchangeable correlation. Removing the leading PC reduces this bias and brings the decomposed FVE estimates close to the true FVE. When population or surrogate eigenvectors are used, the ten-PC panels show that the linear-diagonal residualization is less sensitive to over-removing PCs than the full linear residualization. When sample eigenvectors are used, the red and blue curves overlap because the linear and linear-diagonal residualizations coincide algebraically for sample singular vectors; however, the ten-PC sample-eigenvector estimates are biased upward, consistent with double dipping from estimating the PCs and FVE on the same data. Thus, the surrogate-eigenvector results provide the most relevant practical setting: they avoid double dipping and remain close to the true FVE after the leading strong component is removed.

In the second set of FVE simulations, we fix the use of surrogate eigenvectors and evaluate the full FVE decomposition estimator as the number of removed PCs \(k\) varies under the one-block exchangeable model (Figure~\ref{sim_oneblock_result} and Table~\ref{sim_oneblock_table}). The procedure in Section~\ref{sec:edge_algorithm} selects \(\hat{K}=1\) with probability $0.958$, matching the single strong component in this model. Across the four one-block settings, when choosing \(k=0\), the FVE estimates are biased because the leading exchangeable component violates the weak-correlation conditions required by the high-dimensional estimators. After removing one PC, the total-FVE curves move close to the green true-FVE line for both GWASH (blue line) and LMM-REML (red line).

Table~\ref{sim_oneblock_table} shows the same pattern quantitatively: compared with \(k=0\), using \(k=1\) reduces bias in all settings, and increasing \(k\) from 1 to 10, 50, or 100 changes the bias only slightly. The standard error decreases slightly in LMM-REML but increases for GWASH. This behavior may be explained by the GWASH variance expression in \eqref{gwash_var}: the leading contribution to the variance is the term $\frac{m}{n\mu_2}$, whose denominator becomes small when the decomposition enters the weak-correlation regime. After decomposition, the empirical standard errors are also close to the formula-based standard errors reported in Appendix~\ref{appendix:SE_comparison}. Finally, the adjusted-\(R^2\) component tracks the true low-dimensional FVE well: the first PC captures most of the shared correlation in the one-block model, and additional PCs add only small increments to the low-dimensional component.

\begin{figure}
    \centering
    \includegraphics[width=\linewidth]{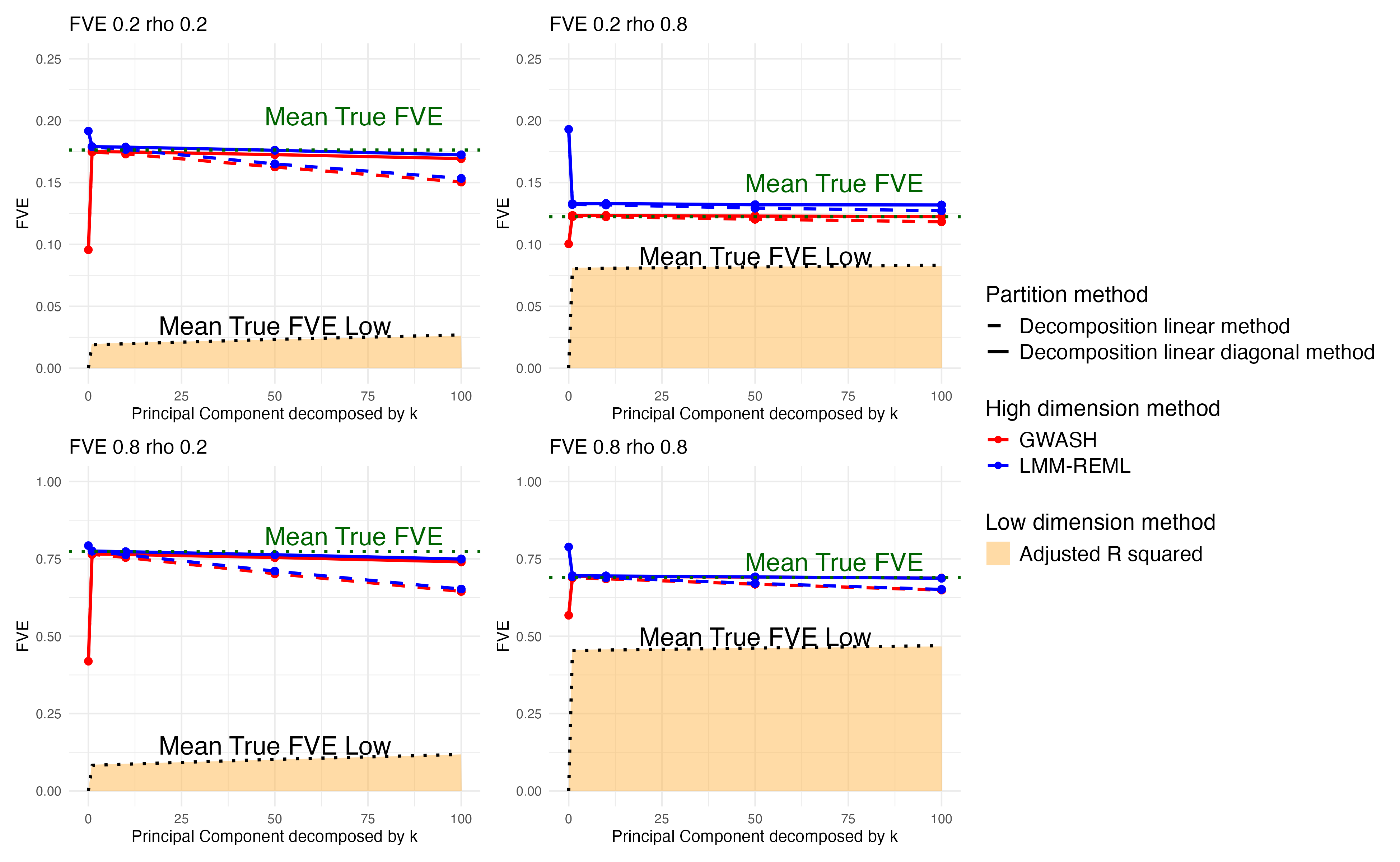}
    \caption{
    Simulation results under a one-block exchangeable correlation structure. The first and second rows correspond to simulations with nominal FVE values of $0.2$ and $0.8$, respectively. The first and second columns represent correlation levels $\rho = 0.2$ and $\rho = 0.8$. The dotted black line indicates the mean of the true FVE for the low-dimensional component, estimated using adjusted $R^2$. The tan shaded curve shows the adjusted-\(R^2\) estimate of the low-dimensional component \(\hat{h}^2_k\). The dashed green line indicates the mean of the total true FVE (low- plus high-dimensional components), where the high-dimensional component is estimated using two different methods: LMM-REML (blue) and GWASH (red).
    }
    \label{sim_oneblock_result}
\end{figure}

\begin{table}
\centering
\caption{Bias and empirical standard error of GWASH and LMM-REML estimators in Monte Carlo simulation under one-block exchangeable correlation structures for varying FVE ($0.2$, $0.8$) and correlation levels ($\rho = 0.2$, $0.8$).}
\label{sim_oneblock_table}
\begin{tabular}{rcccc|cccc}
\toprule
\multirow{3}{*}{$k$}
& \multicolumn{4}{c|}{FVE = 0.2}
& \multicolumn{4}{c}{FVE = 0.8} \\
\cmidrule(lr){2-5} \cmidrule(lr){6-9}
& \multicolumn{2}{c}{GWASH}
& \multicolumn{2}{c|}{LMM-REML}
& \multicolumn{2}{c}{GWASH}
& \multicolumn{2}{c}{LMM-REML} \\
\cmidrule(lr){2-3} \cmidrule(lr){4-5}
\cmidrule(lr){6-7} \cmidrule(lr){8-9}
& Bias & SE
& Bias & SE 
& Bias & SE 
& Bias & SE  \\
\midrule
\multicolumn{9}{l}{Correlation $\boldsymbol{\rho = 0.2}$} \\
0   & -0.081 & 0.043 &  0.014 & 0.072 & -0.355 & 0.084 &  0.019 & 0.039 \\
1   &  0.000 & 0.065 &  0.003 & 0.065 & -0.008 & 0.065 &  0.002 & 0.040 \\
10  &  0.000 & 0.064 &  0.002 & 0.064 & -0.010 & 0.065 & -0.001 & 0.041 \\
50  & -0.002 & 0.064 &  0.000 & 0.064 & -0.019 & 0.064 & -0.010 & 0.041 \\
100 & -0.005 & 0.062 & -0.004 & 0.062 & -0.033 & 0.062 & -0.024 & 0.039 \\
\midrule
\multicolumn{9}{l}{Correlation $\boldsymbol{\rho = 0.8}$} \\
0   & -0.022 & 0.021 &  0.070 & 0.130 & -0.123 & 0.029 &  0.098 & 0.045 \\
1   &  0.002 & 0.060 &  0.010 & 0.050 &  0.003 & 0.041 &  0.005 & 0.037 \\
10  &  0.002 & 0.060 &  0.010 & 0.050 &  0.003 & 0.041 &  0.004 & 0.037 \\
50  &  0.002 & 0.060 &  0.009 & 0.050 &  0.001 & 0.042 &  0.001 & 0.037 \\
100 &  0.001 & 0.058 &  0.009 & 0.049 & -0.002 & 0.041 & -0.003 & 0.036 \\
\bottomrule
\end{tabular}
\end{table}

\begin{table}
\centering
\caption{Bias and empirical standard error (SE) of GWASH and LMM-REML estimators in Monte Carlo simulation under a two-block exchangeable correlation structure ($\rho_1 = 0.1$, $\rho_2 = 0.9$, $\rho_\beta = 0$) for varying FVE ($0.2$, $0.8$).}
\label{sim_twoblock_table}
\begin{tabular}{rcccc|cccc}
\toprule
\multirow{3}{*}{$k$}
& \multicolumn{4}{c|}{FVE = 0.2}
& \multicolumn{4}{c}{FVE = 0.8} \\
\cmidrule(lr){2-5} \cmidrule(lr){6-9}
& \multicolumn{2}{c}{GWASH}
& \multicolumn{2}{c|}{LMM-REML}
& \multicolumn{2}{c}{GWASH}
& \multicolumn{2}{c}{LMM-REML} \\
\cmidrule(lr){2-3} \cmidrule(lr){4-5}
\cmidrule(lr){6-7} \cmidrule(lr){8-9}
& Bias & SE
& Bias & SE
& Bias & SE
& Bias & SE \\
\midrule
\multicolumn{9}{l}{Two-block exchangeable: $\rho_1 = 0.1,\ \rho_2 = 0.9$} \\
0   & -0.067 & 0.025 &  0.017 & 0.084 &  0.301 & 0.043 & -0.047 & 0.028 \\
1   & -0.091 & 0.028 & -0.022 & 0.064 & -0.269 & 0.024 & -0.009 & 0.023 \\
2   & -0.009 & 0.063 & -0.007 & 0.060 &  0.000 & 0.035 & -0.006 & 0.022 \\
10  & -0.010 & 0.063 & -0.008 & 0.060 &  0.000 & 0.035 & -0.007 & 0.021 \\
50  & -0.015 & 0.062 & -0.012 & 0.060 & -0.002 & 0.034 & -0.010 & 0.022 \\
100 & -0.021 & 0.060 & -0.018 & 0.056 & -0.004 & 0.035 & -0.015 & 0.022 \\
\bottomrule
\end{tabular}
\end{table}

Similar trends are observed under the two-block exchangeable correlation structure. Results for this setting are shown in Figure~\ref{sim_twoblock_figure} and Table~\ref{sim_twoblock_table}. The procedure in Section~\ref{sec:edge_algorithm} selects \(\hat{K}=2\) with probability $0.954$. When no PCs are removed, both GWASH and LMM-REML can exhibit non-negligible bias. Removing only one PC does not fulfill the \(\mathrm{WD}_0\) and cannot eliminate the bias. However, once two or more PCs are removed, both methods produce approximately unbiased estimates, and the empirical standard errors remain stable as additional PCs are removed. These results further demonstrate the robustness of the decomposition-based correction across covariance structures beyond the one-block exchangeable setting.

\begin{figure}
\includegraphics[width=\linewidth]{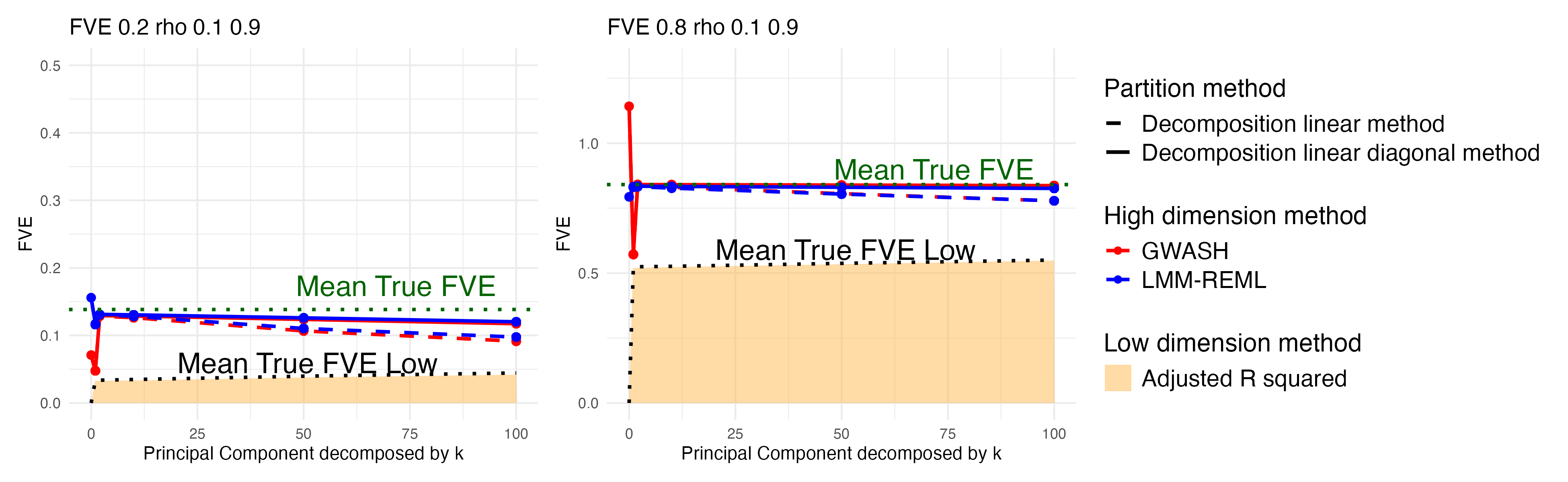}
    \caption{Simulation results under a two-block exchangeable correlation structure with $\rho_1 = 0.1$ and $\rho_2 = 0.9$ with $\rho_B = 0$. The first and second columns represent FVE $0.2$ and $0.8$. The dotted black line indicates the mean of the true FVE for the low-dimensional component, estimated using adjusted $R^2$. The dashed green line indicates the mean of the total true FVE (low- plus high-dimensional components), where the high-dimensional component is estimated using two different methods: LMM-REML (blue) and GWASH (red). The projection for the high dimension method is shown as a solid line for the linear diagonal method and as a dashed line for the linear model.}
\label{sim_twoblock_figure}
\end{figure}

\subsection{Additional simulations} 
\label{sec:addtional_simulation}
We conduct additional simulations investigating the splitting of the surrogate and study samples, as well as a separate simulation comparing empirical and estimated standard errors from Equation~\eqref{eq:var_formula_FVE_decom}. The results show that, as the surrogate sample size increases and becomes comparable to the study sample size, the bias decreases and the estimated FVE approaches the true FVE. The estimator of standard error in Equation~\eqref{eq:var_formula_FVE_decom} also agrees well with the empirical standard error in most settings. One exception occurs when the FVE is high and the correlation is low, in which case the GWASH standard error tends to be underestimated. Further details are provided in Appendix~\ref{surrogate_sample_datasplitting} and Appendix~\ref{appendix:SE_comparison}, respectively.

\section{Data application}
\subsection{ABCD data}

The Adolescent Brain Cognitive Development (ABCD) Study is a large-scale, ongoing 10-year longitudinal study enrolling 11,880 children starting at 9–11 years old and designed to investigate brain development and health from late childhood into early adulthood. The ABCD data repository grows and changes over time. The ABCD data used in this paper came from \href{http://dx.doi.org/10.15154/z563-zd24}{http://dx.doi.org/10.15154/z563-zd24}. In this study, we utilize data from the ABCD Study Release 3.0, specifically structural T1-weighted MRI data, to examine the FVE of cortical surface area in relation to the neuropsychiatric outcome of general cognitive ability, assessed using a ``gold standard'' measure of crystallized intelligence in children (\cite{akshoomoff_viii_2013}).

The MRI data were collected from 21 study sites using harmonized acquisition protocols across Siemens Prisma, GE 750, and Philips 3T scanners. To account for potential site and scanner-related variability, scanner ID was included in all analyses. A comprehensive description of the MRI acquisition protocols is provided in \cite{casey_adolescent_2018}. Our analysis focused on the baseline visit of the ABCD cohort and yielded a final sample of 6,240 subjects with surface area data and complete demographic information (age, sex, race/ethnicity) after quality control procedures.

\subsection{MRI preprocessing}
Preprocessing of the T1-weighted images was performed using a software package developed and maintained in-house at the Center for Multimodal Imaging and Genetics (CMIG) at the University of California, San Diego (UCSD) (\cite{hagler_image_2019}) to extract vertex-wise measures of cortical surface area. The preprocessing pipeline consisted of several key steps: First, images underwent distortion correction, followed by spatial alignment to a standardized atlas space. Automated segmentation was then applied to identify gray matter, white matter, and cerebrospinal fluid. The cortical surface was reconstructed by tessellating the gray/white matter boundary, generating a high-resolution 3D mesh representation of each participant’s cortical surface. To facilitate cross-subject comparisons, a standardized spherical registration approach was employed. Each participant’s cortical surface was non-rigidly aligned to preserve the relative topology of cortical folding patterns (sulci and gyri) across individuals while enabling direct vertex-wise comparisons. Cortical surface area was computed at each vertex in the participant’s native space by summing one-third of the area of all adjacent triangles. These vertex-wise surface area values were then interpolated onto a template, resulting in a dataset containing 20,484 total vertices (10,242 per hemisphere) per participant. After excluding the medial wall and vertices containing no information, each participant’s dataset contained 18,742 vertices.

\subsection{Data analysis result}
\begin{figure}[t!]
    \centering    \includegraphics[width=0.9\linewidth]{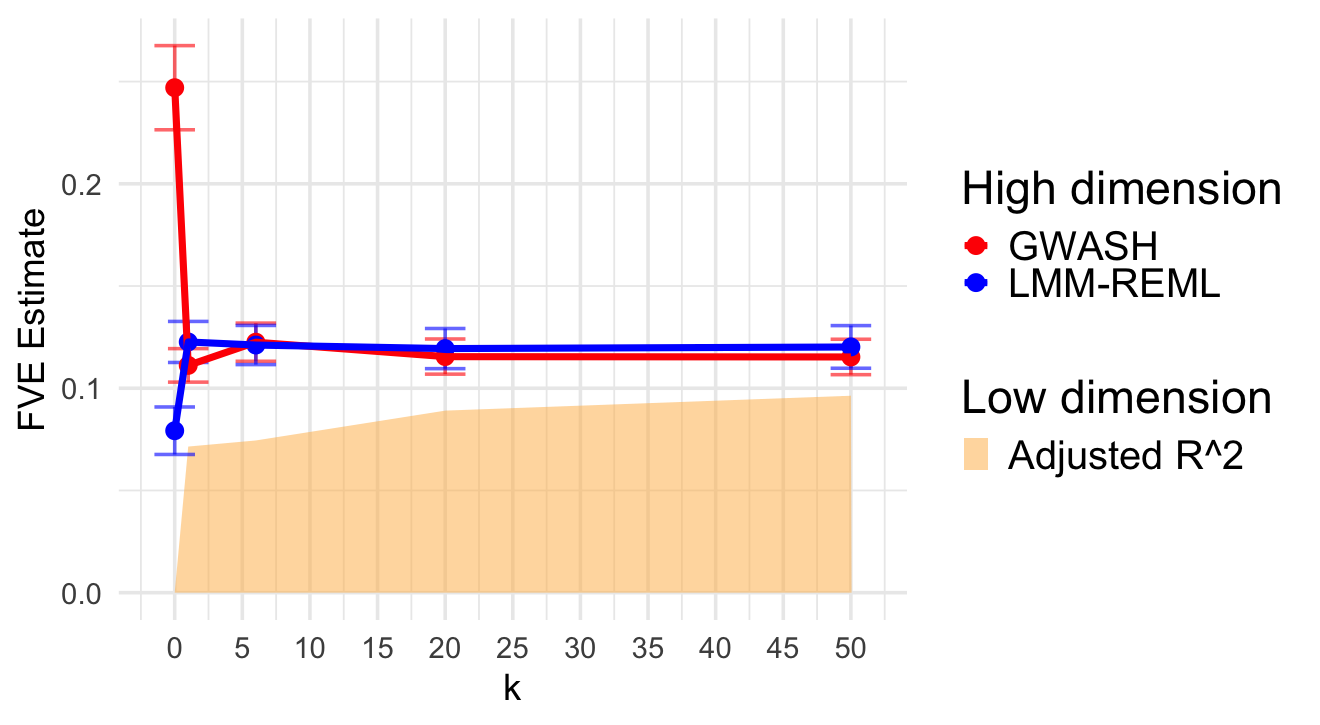}
    \caption{FVE estimation after decomposition for crystallized intelligence using cortical surface area data from the ABCD Study. The low-dimensional component is estimated by adjusted $R^2$, and the residual high-dimensional component is estimated by GWASH (red line) or LMM-REML (blue line) after conducting linear-diagonal residualization schemes in \eqref{left_equation_diagonal}. Error bars show the standard error across 100 permutations generated by randomly splitting the sample and surrogate data into equal-sized groups. }\label{fig:real_data_result}
\end{figure}

We applied the proposed FVE decomposition to cortical surface area data from the ABCD Study with crystallized intelligence as the outcome. Based on the simulation results in Section~\ref{sim_oneblock_result}, we used Ezekiel's adjusted $R^2$ for the low-dimensional domain and GWASH or LMM-REML for the residual high-dimensional component, under linear-diagonal residualization schemes in \eqref{left_equation_diagonal}. The study and surrogate samples were obtained by an equal split, as described in Appendix~\ref{surrogate_sample_datasplitting}. Using the procedure in Section~\ref{sec:edge_algorithm}, we estimate $\hat{K}=6$. We conduct a sensitivity analysis where 1000 permutations are generated by randomly splitting the sample and surrogate data into equal-sized groups. 

Figure~\ref{fig:real_data_result} presents final results on FVE decomposition and estimation using the ABCD data. Direct application of LMM-REML or GWASH without decomposition yields different results between the two methods. As more PCs are incorporated into the lower-dimensional representation, the adjusted $R^2$ values increase as expected. The LMM-REML (blue line) method demonstrates slightly greater stability than GWASH (red), likely due to its use of individual-level data; however, the overall differences between methods remain minimal. When only 1 PC is decomposed, there is still a small gap between these two high dimensional methods. Once we choose $k > \hat{K} = 6$, the two methods reach similar result and the result is relatively stable if the choice of $k$ is not too high. When $k = 6$, the estimated total FVE was approximately $12.1\%$ (s.e. 0.0093 for GWASH and 0.0095 for LMM-REML) with low-dimensional component accounts for around $7.5\%$ and high-dimensional component accounts for $4.6\%$.

These findings highlight the practical importance of estimating FVE through decomposition in the setting where cortical surface area features exhibit strong correlation. The FVE estimate obtained from the proposed decomposition method is substantially larger than estimates suggested by existing analyses based on either marginal associations or non-decomposed high-dimensional approaches. For example, \cite{marek_reproducible_2022} reported the largest univariate FVE of only 2.56\% across all univariate brain-wide regressions in ABCD study, suggesting that univariate associations substantially underestimate the FVE. A related limitation may arise for high-dimensional method such as LMM-REML when they are applied without separating the dominant low-dimensional structure from the residual high-dimensional component. As shown in the UK Biobank analysis, \citet{couvyduchesne_unified_2020} estimated that BLUP-based grey-matter scores derived from surface area explained only 0.59\% of the variation in fluid intelligence under an LMM-REML-based approach. By contrast, our proposed FVE decomposition yields a stable estimate of approximately 12.1\% in the ABCD data, indicating that a non-negligible proportion of variation in crystallized intelligence can be attributed to cortical surface area when both the leading low-dimensional structure and the residual high-dimensional component are properly accounted for.

\section{Discussion}

In this paper, we developed a principal-components decomposition framework for estimating the FVE in high-dimensional linear models with strongly correlated predictors. Existing high-dimensional FVE estimators, including GWASH and LMM-REML, are useful when the predictor covariance structure is sufficiently weakly dependent, but their performance can deteriorate when this condition is violated. Our method directly addresses this issue by separating the predictor space into a low-rank component spanned by the leading PCs and a residual high-dimensional component whose covariance structure is weaker after the dominant directions are removed. The total FVE can then be estimated by combining the low-dimensional FVE, estimated using adjusted $R^2$, with the residual high-dimensional FVE, estimated using GWASH or LMM-REML.

A central message of this work is that strong correlation among predictors is not merely a nuisance feature of neuroimaging data, but a structural feature that can substantially affect FVE estimation. This problem is distinct from ordinary covariate adjustment or demeaning. Regressing out demographic covariates does not necessarily eliminate strong dependence among the imaging features themselves. Removing a global mean (\cite{kriegeskorte_representational_2008, azriel_estimation_2020,  revsine_unifying_2024}) is often used to eliminate this strong dependence. However, the global mean fully captures the strong dependence only in the case of a one-block exchangeable correlation (\cite{azriel_empirical_2015}). More complex covariance structures, such as a two-block correlation structure or correlation structures observed in real data, are represented by a higher number of leading eigenvectors with more specific entries.

Moreover, the leading PCs may carry real signal for the phenotype and should not simply be discarded. The proposed decomposition framework is designed to preserve this low-rank contribution while still allowing standard high-dimensional estimators to be applied to the residual feature space where their assumptions are more appropriate. For example, as seen in our data analysis results, 62\% of the explained variation is attributable to broad low-rank anatomical variation, while the remaining 38\% is attributable to residual high-dimensional spatial patterns.

The proposed method is supported by both theory and simulations that show that, under a one-block or two-block exchangeable correlation structure, removing the appropriate number of dominant PCs produces residuals for which the desired weak dependence assumption holds. The result is stable if more than the appropriate number of dominant PCs is removed, so that exact estimation of this number is not necessary.

The FVE estimation simulations further demonstrate the practical impact of this correction. When no PCs are removed, both GWASH and LMM-REML can exhibit non-negligible bias. Removing too few PCs may leave part of the strong covariance structure in the residual design, so the resulting estimator can remain biased. Once an adequate number of PCs is removed, however, both GWASH and LMM-REML produce approximately unbiased estimates. Following the procedure in Section~\ref{sec:edge_algorithm} we can estimate the $\hat K$, but our simulation shows that the estimates remain stable when more PCs are removed than strictly necessary, indicating that mild over-decomposition is less harmful than under-decomposition in the settings considered here. This robustness is especially useful in practice, where the true number of strong components is unknown and must be estimated from data. In this work we propose Spectrode \citep{dobriban_efficient_2015} as a method for estimating this number, however the robustness makes it so that the estimate does not need to be precise and can be treated as a lower bound.

Our work also highlights the importance of using surrogate eigenvectors. If the same study sample is used both to estimate the PCs and to estimate FVE after residualization, the procedure can suffer from double dipping \citep{kriegeskorte_circular_2009, ball_double_2020}. Sample eigenvectors may absorb sample-specific noise, and residualizing with respect to these directions can distort the remaining covariance structure. In contrast, surrogate eigenvectors estimated from an independent but representative reference sample break this feedback loop. Conditional on the surrogate sample, the eigenvectors are fixed before FVE estimation is performed in the study sample. Our simulations show that surrogate eigenvectors lead to more stable and approximately unbiased FVE estimates. This feature makes the method conceptually similar to reference-panel approaches used in genetic analyses \citep{schizophrenia_working_group_of_the_psychiatric_genomics_consortium_ld_2015, taliun_laser_2017}, while adapting the idea to neuroimaging covariance decomposition.

The real-data application to the ABCD cortical surface area data illustrates how the method can be used in a large-scale neuroimaging study. In the analysis of crystallized intelligence, direct application of GWASH and LMM-REML without decomposition produced very different estimates for the two methods. After decomposing the FVE with at least $\hat{k} = 6$ leading PCs, the two high-dimensional methods gave comparable and more stable results with respect to $k$. The estimated total FVE was approximately $12.1\%$, with the low-dimensional component accounting for $7.5\%$ and the high-dimensional component accounting for $4.6\%$, suggesting that the association of cortical surface area with crystallized intelligence is explained by both global and local components. Our total FVE estimate of $12.1\% \pm 0.1\%$ is substantially higher than previous estimates such as  univariate FVE of only 2.56\% in ABCD data \citep{marek_reproducible_2022} and 0.59\% in UK Biobank data \citep{couvyduchesne_unified_2020}. Our results indicate a stronger association than previously known.

Several limitations remain for this work. First, the use of surrogate eigenvectors introduces a tradeoff. Independent surrogate data help avoid double dipping, but splitting the available sample into study and surrogate sets reduces the effective sample size for FVE estimation. If the surrogate sample is too small, the estimated eigenvectors may be noisy; if the surrogate sample comes from a different population, scanner distribution, or preprocessing pipeline, covariance mismatch may introduce additional bias. Here we only investigate the sample-splitting strategies empirically.

Second, the eigenvector estimation of the covariance matrix is time-consuming when the dimension of the covariates is large. This problem is ameliorated by the fact that only a small number of leading eigenvectors need to be computed (in our case, only 6). Also, as commonly done in genetics \citep{chen_improved_2013, taliun_laser_2017}, a set of surrogate eigenvectors could be computed once in a surrogate dataset and stored for future use without additional computation. Third, our artificial simulations included one- and two-block exchangeable correlation structures, as well as realistic simulations. We expect other strong correlations structures to behave similarly. 

Future work may include theoretical investigation of optimal sample-splitting strategies and other strong correlations structures, as well as extensions to regression settings beyond linear such as logistic with binary outcomes.


\section{Code availability}

The implementation code used for this paper can be found at https://github.com/Aileen-Luo/FVE-decomposition-with-strong-correlation.

\section{Author contributions statement}

ML: Conceptualization, Methodology, Software, Formal Analysis, Writing - Review and Editing, Writing Original Draft
CC: Conceptualization, Methodology, Formal Analysis, Writing - Review and Editing, Supervision
DA: Conceptualization, Methodology, Formal Analysis, Writing - Review and Editing, Supervision
AS: Conceptualization, Methodology, Formal Analysis, Writing - Review and Editing, Project Administration, Funding Acquisition.

\section{Acknowledgment}

This work was partially supported by National Institutes of Health [R01MH128923].

\addcontentsline{toc}{chapter}{Bibliography}
\bibliography{Bibliography} 
\bibliographystyle{apa}

\appendix
\counterwithin{figure}{section}

\section{Appendix}

\subsection{Proof of proposition ~\ref{prop:pop-AS}}
\begin{proof}
Since $V_k$ collects the top $k$ orthonormal eigenvectors of $\Sigma$,
the projection $(I-V_kV_k^\top)$ removes the first $k$ eigendirections of
$\Sigma$. Hence the nonzero eigenvalues of $\Sigma_{\cdot k}$ are
$\lambda_{k+1},\lambda_{k+2},\ldots,\lambda_{m}.$
Because $k\ge K$, all eigenvalues remaining after projection are bulk
eigenvalues. By the bounded-bulk assumption \eqref{eq:bulk_bound}, all nonzero eigenvalues of $\Sigma_{\cdot k}$ are bounded
above by $C$.

By assumption \eqref{eq:residual_variances_bounded}, since $c_k$ is a lower bound on the diagonal elements of $D_{\cdot k}$ (which is a diagonal matrix),
\[
    D_{\cdot k}\ge c_kI_m \qquad \Rightarrow  \qquad
    D_{\cdot k}^{-1}
    \le
    \frac{1}{c_k}I_m.
\]
entry-wise. As a result, the re-standardized residual correlation matrix has bounded trace
\[
\tr\!\left(
    \widetilde{\Sigma}_{\cdot k}^{\,2}
    \right)
  =  \tr\!\left(
    D_{\cdot k}^{-1/2}
    \Sigma_{\cdot k}
    D_{\cdot k}^{-1} 
    \Sigma_{\cdot k}
    D_{\cdot k}^{-1/2}
    \right)=   \tr\!\left( \left[D_{\cdot k}^{-1}
    \Sigma_{\cdot k} \right]^2
    \right) \le \frac{1}{c_k^2} \tr\!\left(\Sigma_{\cdot k}^{\,2}
    \right).
\]
Since the nonzero eigenvalues of $\Sigma_{\cdot k}$ are bounded above by
$C$, 
\[
    \tr\!\left(
    \widetilde{\Sigma}_{\cdot k}^{\,2}
    \right)
    \le
    m
    \left(
    \frac{C}{c_k}
    \right)^2.
\]
Dividing by $m^2$ gives
\[
    \frac{1}{m^2}
    \tr\!\left(
    \widetilde{\Sigma}_{\cdot k}^{\,2}
    \right)
    \le
    \frac{1}{m}
    \left(
    \frac{C}{c_k}
    \right)^2
    \longrightarrow 0.
\]
\end{proof}

\subsection{Proof of Proposition ~\ref{prop:pc_fve_decomp}}

\begin{proof}

By the PC reparameterization in Equation~\eqref{eq:predictor_decomposition}, $\vec{\x}\bbeta=\vec{\w}_k\bgamma_k+\vec{\w}_{.k}\bgamma_{.k}$,
so
\[
h_\beta^2
=\frac{\Var(\vec{\x}\bbeta)}{\Var(y)}
=\frac{\Var(\vec{\w}_k\bgamma_k+\vec{\w}_{.k}\bgamma_{.k})}{\Var(y)}.
\]
Since $V$ is an orthonormal eigenbasis of $\Sigma$, we have
\[
\E(\vec{\w}_k^\top \vec{\w}_{.k})
=\E\!\bigl(V_k^\top \vec{\x}^\top \vec{\x} V_{.k}\bigr)
= V_k^\top \Sigma V_{.k}
= 0,
\]
which implies $\Cov(\vec{\w}_k\bgamma_k,\vec{\w}_{.k}\bgamma_{.k})=0$ and hence
\[
\Var(\vec{\x}\bbeta)=\Var(\vec{\w}_k\bgamma_k)+\Var(\vec{\w}_{.k}\bgamma_{.k}).
\]
Dividing by $\Var(y)$ yields
\[
h_\beta^2
=\frac{\Var(\vec{\w}_k\bgamma_k)}{\Var(y)}
+\frac{\Var(\vec{\w}_{.k}\bgamma_{.k})}{\Var(y)}
= h_k^2 + \frac{\Var(\vec{\w}_{.k}\bgamma_{.k})}{\Var(y)},
\]
where the first term is $h_k^2$ by \eqref{eq:hk_def}.

Now consider the residualized model \eqref{eqform2}: $y_{.k}=\vec{\w}_{.k}\bgamma_{.k}+\epsilon$.
By definition of $y_{.k}$, we also have $y=\vec{\w}_k\bgamma_k+y_{.k}$. Using again
$\Cov(\vec{\w}_k\bgamma_k, \vec{\w}_{.k}\bgamma_{.k})=0$ and $\Cov(\vec{\w}_k\bgamma_k,\epsilon)=0$,
it follows that $\Cov(\vec{\w}_k\bgamma_k,y_{.k})=0$ and therefore
\[
\Var(y)=\Var(\vec{\w}_k\bgamma_k)+\Var(y_{.k})
\quad\Rightarrow\quad
\frac{\Var(y_{.k})}{\Var(y)}=1-\frac{\Var(\vec{\w}_k\bgamma_k)}{\Var(y)}=1-h_k^2.
\]
Finally, by the definition \eqref{eq:hdotk_def},
\[
\frac{\Var(\vec{\w}_{.k}\bgamma_{.k})}{\Var(y)}
=
\frac{\Var(y_{.k})}{\Var(y)}\cdot
\frac{\Var(\vec{\w}_{.k}\gamma_{.k})}{\Var(y_{.k})}
=
(1-h_k^2)\,h_{.k}^2.
\]
Substituting this into the earlier expression for $h_\beta^2$ gives
$h_\beta^2 = h_k^2 + (1-h_k^2)h_{.k}^2$, proving \eqref{eq:pc_fve_identity}.
\end{proof}

\subsection{Derivations for the examples in Section~\ref{section_example_derivation}}
\label{appendix:derivation_for_example}
This appendix gives detailed derivation for one-block and two-block exchangeable correlation examples in Section~\ref{section_example_derivation}.
\paragraph{One-block exchangeable correlation}

First we show that this case violates the weak dependence condition. Before any PC is removed, according to Equation~\eqref{eq:weak_cor} we have
\[
\frac{1}{m^2}\tr(\widetilde{\Sigma}^2) =
    \frac{1}{m^2}\tr(\Sigma^2) =
    \frac{1}{m^2}
    \left\{
    \bigl[1+(m-1)\rho\bigr]^2
    +(m-1)(1-\rho)^2
    \right\}.
\]
Taking the limit as $m\to\infty$, we get
\[
    \frac{1}{m^2}\tr(\widetilde{\Sigma}^2)
    \longrightarrow
    \rho^2.
\]
Hence \(\mathrm{WD}_0\) fails if no  PC is removed when \(\rho>0\).

Next we show that after 1 PC is removed, the weak dependence condition is fulfilled. When removing 1 PC, the residual covariance is
\[
\begin{aligned}
    \Sigma_{\cdot 1}
    &=
    (I_m-V_1V_1^\top)\Sigma(I_m-V_1V_1^\top) \\
    &=
    \left(I_m-\frac{1}{m}\mathbf 1_m\mathbf 1_m^\top\right)
    \left\{(1-\rho)I_m+\rho\mathbf 1_m\mathbf 1_m^\top\right\}
    \left(I_m-\frac{1}{m}\mathbf 1_m\mathbf 1_m^\top\right) \\
&=
    (1-\rho)
    \left(I_m-\frac{1}{m}\mathbf 1_m\mathbf 1_m^\top\right).
\end{aligned}
\]
Thus
\[
    D_{\cdot 1}
    :=
    \diag(\Sigma_{\cdot 1})
    =
    (1-\rho)\left(1-\frac{1}{m}\right)I_m
\]
and
\[
\begin{aligned}
    \widetilde{\Sigma}_{\cdot 1}
    &:=
    D_{\cdot 1}^{-1/2}\Sigma_{\cdot 1}D_{\cdot 1}^{-1/2} \\
    &=
    \frac{1}{(1-\rho)(1-1/m)}
    (1-\rho)
    \left(I_m-\frac{1}{m}\mathbf 1_m\mathbf 1_m^\top\right) \\
    &=
    \frac{m}{m-1}
    \left(I_m-\frac{1}{m}\mathbf 1_m\mathbf 1_m^\top\right).
\end{aligned}
\]
The eigenvalues of \(\widetilde{\Sigma}_{\cdot 1}\) are \(0\)  and \(\frac{m}{m-1}\) with multiplicity \(m-1\).
Therefore
\[
\begin{aligned}
    \tr\!\left(\widetilde{\Sigma}_{\cdot 1}^{\,2}\right)
    =
    (m-1)\left(\frac{m}{m-1}\right)^2 
    =
    \frac{m^2}{m-1},
\end{aligned}
\]
so that
\[
\begin{aligned}
    \mu_2
    =
    \frac{1}{m}\tr\!\left(\widetilde{\Sigma}_{\cdot 1}^{\,2}\right) =
    \frac{1}{m}\cdot \frac{m^2}{m-1} =
    \frac{m}{m-1}
    \longrightarrow 1.
\end{aligned}
\]
Also according to Proposition~\ref{prop:pop-AS},
\[
\begin{aligned}
    \frac{1}{m^2}\tr\!\left(\widetilde{\Sigma}_{\cdot 1}^{\,2}\right)
    =
    \frac{1}{m^2}\cdot \frac{m^2}{m-1} =
    \frac{1}{m-1}
    \longrightarrow 0.
\end{aligned}
\]

\paragraph{Two-block exchangeable correlation}

First we show that the two-block exchangeable correlation case violates the weak dependence condition. 
Before any PC is removed, we have
\[
    \Sigma=
    \begin{bmatrix}
    (1-\rho_1)I_{m_1}+\rho_1\mathbf 1_{m_1}\mathbf 1_{m_1}^\top & \rho_B\mathbf 1_{m_1}\mathbf 1_{m_2}^\top\\
    (\rho_B\mathbf 1_{m_1}\mathbf 1_{m_2}^\top)^\top & (1-\rho_2)I_{m_2}+\rho_2\mathbf 1_{m_2}\mathbf 1_{m_2}^\top
    \end{bmatrix},
\]
and
\[
    \tr(\Sigma^2)
    =
    \tr(((1-\rho_1)I_{m_1}+\rho_1\mathbf 1_{m_1}\mathbf 1_{m_1}^\top)^2)+2\tr((\rho_B\mathbf 1_{m_1}\mathbf 1_{m_2}^\top)(\rho_B\mathbf 1_{m_1}\mathbf 1_{m_2}^\top)^\top)+\tr(((1-\rho_2)I_{m_2}+\rho_2\mathbf 1_{m_2}\mathbf 1_{m_2}^\top)^2).
\]
Since \((\mathbf 1_{m_1}\mathbf 1_{m_1}^\top)^2=m_1(\mathbf 1_{m_1}\mathbf 1_{m_1}^\top)\) and \(\tr(\mathbf 1_{m_1}\mathbf 1_{m_1}^\top)=m_1\),
\[
\begin{aligned}
    \tr(\widetilde{\Sigma}^2)
    &=
    \tr(\Sigma^2) =
    m
    +m_1(m_1-1)\rho_1^2
    +m_2(m_2-1)\rho_2^2
    +2m_1m_2\rho_B^2.
\end{aligned}
\]
If \(m_j/m\to\pi_j\) for \(j=1,2\), according to Equation\eqref{eq:weak_cor}
\[
\frac{1}{m^2}\tr(\widetilde{\Sigma}^2)
=
\frac{1}{m}
+
\frac{m_1(m_1-1)}{m^2}\rho_1^2
+
\frac{m_2(m_2-1)}{m^2}\rho_2^2
+2\frac{m_1m_2}{m^2}\rho_B^2 \longrightarrow
    \pi_1^2\rho_1^2+
    \pi_2^2\rho_2^2+
    2\pi_1\pi_2\rho_B^2.
\]
Hence \(\mathrm{WD}_0\) fails if no  PC is removed when \(\rho_1>0\) or \(\rho_2>0\).

Next we show that after 2 PCs are removed, the weak dependence condition is fulfilled.
Assume \(\pi_1,\pi_2>0\) and \(\rho_B^2<\rho_1\rho_2\). Then both eigenvalues \(\lambda_+(M_m)\) and \(\lambda_-(M_m)\) are of order \(m\), so \(K=2\). According to the orthonormal basis \(S\) in Equation~\eqref{eq:eigenspace_two_block}, and the corresponding projection is
\[
    V_2V_2^\top
    =
    u_1u_1^\top+u_2u_2^\top
    =
    \begin{bmatrix}
    \frac{1}{m_1}\mathbf 1_{m_1}\mathbf 1_{m_1}^\top & 0\\[4pt]
    0 & \frac{1}{m_2}\mathbf 1_{m_2}\mathbf 1_{m_2}^\top
    \end{bmatrix}.
\]
Therefore
\[
\begin{aligned}
    \Sigma_{\cdot 2}
    &=
    (I_m-V_2V_2^\top)\Sigma(I_m-V_2V_2^\top) \\
    &=
    \begin{bmatrix}
    \left(I_{m_1}-\frac{1}{m_1}\mathbf 1_{m_1}\mathbf 1_{m_1}^\top\right) & 0\\[4pt]
    0 & \left(I_{m_2}-\frac{1}{m_2}\mathbf 1_{m_2}\mathbf 1_{m_2}^\top\right)
    \end{bmatrix}
    \Sigma
    \begin{bmatrix}
    \left(I_{m_1}-\frac{1}{m_1}\mathbf 1_{m_1}\mathbf 1_{m_1}^\top\right) & 0\\[4pt]
    0 & \left(I_{m_2}-\frac{1}{m_2}\mathbf 1_{m_2}\mathbf 1_{m_2}^\top\right)
    \end{bmatrix}.
\end{aligned}
\]
For the first diagonal block,
\[
\begin{aligned}
\left(I_{m_1}-\frac{1}{m_1}\mathbf 1_{m_1}\mathbf 1_{m_1}^\top\right)
\left\{(1-\rho_1)I_{m_1}+\rho_1\mathbf 1_{m_1}\mathbf 1_{m_1}^\top\right\}
\left(I_{m_1}-\frac{1}{m_1}\mathbf 1_{m_1}\mathbf 1_{m_1}^\top\right) 
=
(1-\rho_1)
\left(I_{m_1}-\frac{1}{m_1}\mathbf 1_{m_1}\mathbf 1_{m_1}^\top\right).
\end{aligned}
\]Similarly for the second diagonal block. The off-diagonal blocks vanish because
\begin{equation*}
    \left(I_{m_1}-\frac{1}{m_1}\mathbf 1_{m_1}\mathbf 1_{m_1}^\top\right)\mathbf 1_{m_1}=\left(I_{m_2}-\frac{1}{m_2}\mathbf 1_{m_2}\mathbf 1_{m_2}^\top\right)\mathbf 1_{m_2}=0.
\end{equation*} Thus the residual covariance matrix is
\[
\Sigma_{\cdot 2}
=
\begin{bmatrix}
(1-\rho_1)
\left(I_{m_1}-\frac{1}{m_1}\mathbf 1_{m_1}\mathbf 1_{m_1}^\top\right)
& 0\\[6pt]
0 &
(1-\rho_2)
\left(I_{m_2}-\frac{1}{m_2}\mathbf 1_{m_2}\mathbf 1_{m_2}^\top\right)
\end{bmatrix}.
\]
with diagonal
\[
D_{\cdot 2}
:=
\diag(\Sigma_{\cdot 2})
=
\begin{bmatrix}
(1-\rho_1)\left(1-\frac{1}{m_1}\right)I_{m_1} & 0\\[4pt]
0 & (1-\rho_2)\left(1-\frac{1}{m_2}\right)I_{m_2}
\end{bmatrix}.
\]
Therefore
\[
\begin{aligned}
\widetilde{\Sigma}_{\cdot 2}
&:=
D_{\cdot 2}^{-1/2}\Sigma_{\cdot 2}D_{\cdot 2}^{-1/2} 
\begin{bmatrix}
\frac{m_1}{m_1-1}
\left(I_{m_1}-\frac{1}{m_1}\mathbf 1_{m_1}\mathbf 1_{m_1}^\top\right)
&0\\[6pt]
0&
\frac{m_2}{m_2-1}
\left(I_{m_2}-\frac{1}{m_2}\mathbf 1_{m_2}\mathbf 1_{m_2}^\top\right)
\end{bmatrix}.
\end{aligned}
\]
The nonzero eigenvalues of \(\widetilde{\Sigma}_{\cdot 2}\) are
\begin{equation*}
    \frac{m_1}{m_1-1}
    \quad\text{with multiplicity }m_1-1\quad\text{and}\quad \frac{m_2}{m_2-1}
    \quad\text{with multiplicity }m_2-1.
\end{equation*}Hence
\[
\begin{aligned}
    \tr\!\left(\widetilde{\Sigma}_{\cdot 2}^{\,2}\right)
    &=
    (m_1-1)\left(\frac{m_1}{m_1-1}\right)^2
    +(m_2-1)\left(\frac{m_2}{m_2-1}\right)^2 =
     m+2+\frac{1}{m_1-1}+\frac{1}{m_2-1}.
\end{aligned}
\]

If \(m_j/m\to\pi_j\in(0,1)\) and \(m_1,m_2\to\infty\), then
\[
    \mu_2
    =
    \frac{1}{m}
    \tr\!\left(\widetilde{\Sigma}_{\cdot 2}^{\,2}\right)
    =
    1+\frac{2}{m}+o\!\left(\frac{1}{m}\right)
    \longrightarrow 1.
\]
So following Proposition~\ref{prop:pop-AS},
\[
\begin{aligned}
    \frac{1}{m^2}
    \tr\!\left(\widetilde{\Sigma}_{\cdot 2}^{\,2}\right)
    =
    \frac{m+2+o(1)}{m^2} =
    O\!\left(\frac{1}{m}\right)
    \longrightarrow 0.
\end{aligned}
\]

\subsection{Delta-method variance for the decomposed FVE estimator}
\label{appendix:delta_var_decomp}

Recall from Proposition~\ref{prop:pc_fve_decomp} that
\[
h^2_{\beta}=h^2_k+(1-h^2_k)h^2_{.k}.
\]
Accordingly, the estimator is
\[
\hat{h}^2_{\beta}
=
\hat{h}^2_k+(1-\hat{h}^2_k)\hat{h}^2_{.k}.
\]
Define
\[
g(a,b)=a+(1-a)b=a+b-ab,
\]
so that
\[
h^2_{\beta}=g(h^2_k,h^2_{.k}),
\qquad
\hat{h}^2_{\beta}=g(\hat{h}^2_k,\hat{h}^2_{.k}).
\]

Let
\[
\Delta_k=\hat{h}^2_k-h^2_k,
\qquad
\Delta_{.k}=\hat{h}^2_{.k}-h^2_{.k}.
\]
Then
\begin{align*}
\hat{h}^2_{\beta}-h^2_{\beta}
&=
g(h^2_k+\Delta_k,\;h^2_{.k}+\Delta_{.k})-g(h^2_k,h^2_{.k}) \\
&=
(h^2_k+\Delta_k)+(1-h^2_k-\Delta_k)(h^2_{.k}+\Delta_{.k})
-\bigl[h^2_k+(1-h^2_k)h^2_{.k}\bigr] \\
&=
(1-h^2_{.k})\Delta_k
+
(1-h^2_k)\Delta_{.k}
-
\Delta_k\Delta_{.k}.
\end{align*}
Thus the first-order Taylor expansion is
\[
\hat{h}^2_{\beta}-h^2_{\beta}
=
(1-h^2_{.k})(\hat{h}^2_k-h^2_k)
+
(1-h^2_k)(\hat{h}^2_{.k}-h^2_{.k})
+
R_n,
\]
with exact remainder
\[
R_n
=
-(\hat{h}^2_k-h^2_k)(\hat{h}^2_{.k}-h^2_{.k}).
\]

If both component estimators are root-$n$ consistent, that is,
\[
\hat{h}^2_k-h^2_k=O_p(n^{-1/2}),
\qquad
\hat{h}^2_{.k}-h^2_{.k}=O_p(n^{-1/2}),
\]
then
\[
R_n=O_p(n^{-1}),
\]
so the product term is of smaller order than the linear terms. Therefore,
\[
\hat{h}^2_{\beta}-h^2_{\beta}
=
(1-h^2_{.k})(\hat{h}^2_k-h^2_k)
+
(1-h^2_k)(\hat{h}^2_{.k}-h^2_{.k})
+
o_p(n^{-1/2}).
\]

Taking variances on both sides gives the first-order approximation
\begin{align*}
\Var(\hat{h}^2_{\beta})
&=
\Var\!\left[
(1-h^2_{.k})(\hat{h}^2_k-h^2_k)
+
(1-h^2_k)(\hat{h}^2_{.k}-h^2_{.k})
\right]
+
o(n^{-1}) \\
&=
(1-h^2_{.k})^2\Var(\hat{h}^2_k)
+
(1-h^2_k)^2\Var(\hat{h}^2_{.k}) \\
&\qquad
+
2(1-h^2_{.k})(1-h^2_k)\Cov(\hat{h}^2_k,\hat{h}^2_{.k})
+
o(n^{-1}).
\end{align*}

In practice, the unknown weights $1-h^2_k$ and $1-h^2_{.k}$ are replaced by their
plug-in estimators, leading to
\[
\Var(\hat{h}^2_{\beta})
\approx
(1-\hat{h}^2_{.k})^2 \Var(\hat{h}^2_k)
+
(1-\hat{h}^2_k)^2 \Var(\hat{h}^2_{.k})
+
2(1-\hat{h}^2_{.k})(1-\hat{h}^2_k)\Cov(\hat{h}^2_k,\hat{h}^2_{.k}).
\]
We use this estimator to 
approximately estimate the variance:
\[
\Var(\hat{h}^2_{\beta})
\approx
(1-\hat{h}^2_{.k})^2 \Var(\hat{h}^2_k)
+
(1-\hat{h}^2_k)^2 \Var(\hat{h}^2_{.k}).
\] 

\subsection{Simulation of splitting surrogate and sample data}
\label{surrogate_sample_datasplitting}

We conduct a simulation following the procedure in Section~\ref{sec_FVE_decomposition} to evaluate how the decomposition estimator depends on the size of the surrogate dataset used to estimate eigenvectors. We fix the total number of observations across the study sample and the surrogate sample at $n_{\mathrm{tot}}=2000$, with $m=4000$ predictors, and vary the split between the study sample size $n$ and the surrogate sample size $n_{\mathrm{sur}}=n_{\mathrm{tot}}-n$. The surrogate data are used only to compute the surrogate eigenvectors $V_{\mathrm{sur},k}$, which define the PC scores $W_k=XV_{\mathrm{sur},k}$. We then estimate the low-dimensional component $h_k^2$ using Ezekiel’s adjusted $R^2$, and estimate the high-dimensional residual component $h_{\cdot k}^2$ using GWASH or LMM-REML applied to the residualized data, where the response is residualized using the left-diagonal projection in \eqref{left_equation_diagonal}. The two components are combined to obtain the total FVE estimate as described in Section~\ref{sec_FVE_decomposition}.

\begin{figure}
    \centering
    \includegraphics[width=0.8\linewidth]{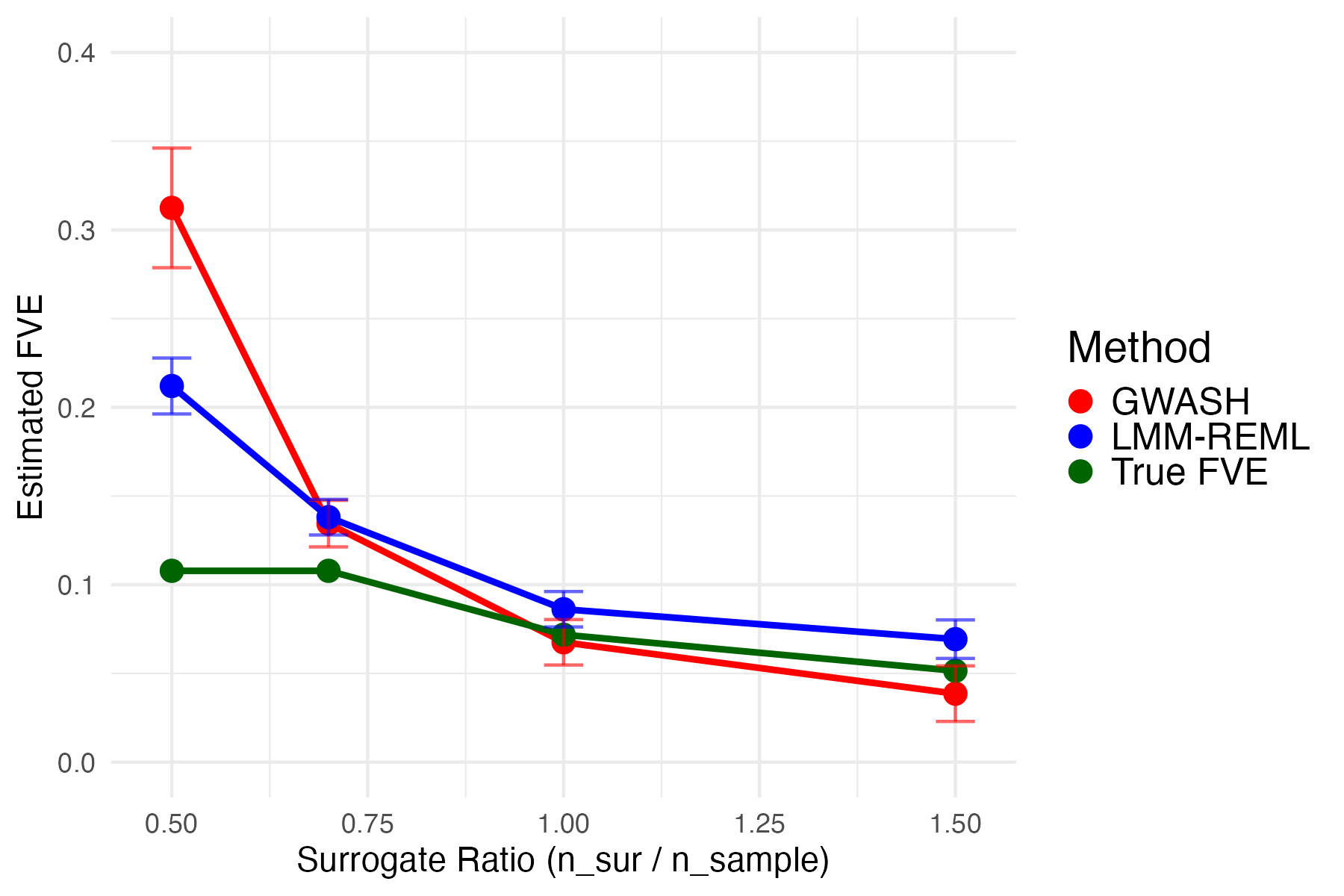}
    \caption{Monte Carlo stimulation for splitting surrogate and sample data under one-block exchangeable correlation structure. After conducting FVE decomposition with linear diagonal projection on 1 PC, the high-dimension FVE is estimated by GWASH (solid red line) and LMM-REML (solid blue line) and compared with the true FVE (solid green line). Error bars show the standard error across 100 replications.}
    \label{fig:ratio_between_surrogateandsample}
\end{figure}

Figure~\ref{fig:ratio_between_surrogateandsample} shows that when the surrogate sample is substantially smaller than the study sample (e.g., $n_{\mathrm{sur}}$ is about half of $n$), the resulting FVE estimates exhibit noticeable bias, reflecting increased noise in the surrogate eigenvectors. As the surrogate sample size increases and becomes comparable to the study sample size, this bias decreases and the estimated FVE approaches the true FVE.

\subsection{Empirical standard error compared with standard error estimator}
\label{appendix:SE_comparison}
We further evaluate the standard error estimator for the decomposed FVE estimator in Equation~\eqref{eq:var_formula_FVE_decom}. In this section, the empirical standard errors are calculated as the Monte Carlo standard deviation of the FVE estimates across simulation replicates. The formula-based SE is based on the first-order delta-method approximation for the decomposed estimator
$\hat h^2_\beta$
via Equation~\eqref{eq:var_formula_FVE_decom}. The simulation results below show that this approximation performs well in most settings, especially after the dominant PCs have been removed.

Table~\ref{sim_oneblock_table_SEF_SEE} reports the comparison under the one-block exchangeable correlation structure. Overall, the formula-based and empirical SEs are close across FVE levels, correlation strengths, and choices of $k$. In the weaker correlation setting $(\rho=0.2)$, the agreement is strong for both GWASH and LMM-REML, regardless of whether no PC or multiple PCs are removed. In the stronger correlation setting $(\rho=0.8)$, the agreement remains good for LMM-REML. Table~\ref{sim_twoblock_table_SE_only} gives the analogous comparison for the two-block exchangeable correlation structure. Taken together, these results suggest that the proposed standard error formula provides a reasonable approximation to the Monte Carlo variability of the decomposed FVE estimator. 
\begin{table}
\centering
\caption{Formula-based standard error and empirical standard error  of GWASH and LMM-REML estimators in Monte Carlo simulation under one-block exchangeable correlation structures for varying FVE ($0.2$, $0.8$) and correlation levels ($\rho = 0.2$, $0.8$).}
\label{sim_oneblock_table_SEF_SEE}
\resizebox{\textwidth}{!}{%
\begin{tabular}{rcccc|cccc}
\toprule
\multirow{3}{*}{$k$}
& \multicolumn{4}{c|}{FVE = 0.2}
& \multicolumn{4}{c}{FVE = 0.8} \\
\cmidrule(lr){2-5} \cmidrule(lr){6-9}
& \multicolumn{2}{c}{GWASH}
& \multicolumn{2}{c|}{LMM-REML}
& \multicolumn{2}{c}{GWASH}
& \multicolumn{2}{c}{LMM-REML} \\
\cmidrule(lr){2-3} \cmidrule(lr){4-5}
\cmidrule(lr){6-7} \cmidrule(lr){8-9}
& Formula-based & Empirical
& Formula-based & Empirical
& Formula-based & Empirical
& Formula-based & Empirical \\
\midrule
\multicolumn{9}{l}{Correlation $\boldsymbol{\rho = 0.2}$} \\
0   & 0.042 & 0.044 & 0.073 & 0.075 & 0.089 & 0.084 & 0.041 & 0.039 \\
1   & 0.067 & 0.066 & 0.065 & 0.067 & 0.070 & 0.065 & 0.043 & 0.040 \\
10  & 0.067 & 0.066 & 0.065 & 0.067 & 0.070 & 0.065 & 0.043 & 0.041 \\
50  & 0.067 & 0.065 & 0.065 & 0.066 & 0.068 & 0.064 & 0.043 & 0.041 \\
100 & 0.066 & 0.063 & 0.065 & 0.064 & 0.067 & 0.062 & 0.042 & 0.039 \\
\midrule
\multicolumn{9}{l}{Correlation $\boldsymbol{\rho = 0.8}$} \\
0   & 0.022 & 0.021 & 0.136 & 0.131 & 0.047 & 0.029 & 0.043 & 0.045 \\
1   & 0.061 & 0.062 & 0.062 & 0.051 & 0.042 & 0.041 & 0.037 & 0.037 \\
10  & 0.061 & 0.062 & 0.062 & 0.051 & 0.042 & 0.041 & 0.037 & 0.037 \\
50  & 0.061 & 0.061 & 0.062 & 0.052 & 0.041 & 0.042 & 0.036 & 0.037 \\
100 & 0.061 & 0.060 & 0.062 & 0.050 & 0.040 & 0.041 & 0.035 & 0.036 \\
\bottomrule
\end{tabular}}
\end{table}

\begin{table}
\centering
\caption{Formula-based standard error and empirical standard error of GWASH and LMM-REML estimators in Monte Carlo simulation under a two-block exchangeable correlation structure ($\rho_1 = 0.1$, $\rho_2 = 0.9$, $\rho_\beta = 0$) for varying FVE ($0.2$, $0.8$).}
\label{sim_twoblock_table_SE_only}

\resizebox{\textwidth}{!}{%
\begin{tabular}{rcccc|cccc}
\toprule
\multirow{3}{*}{$k$}
& \multicolumn{4}{c|}{FVE = 0.2}
& \multicolumn{4}{c}{FVE = 0.8} \\
\cmidrule(lr){2-5} \cmidrule(lr){6-9}
& \multicolumn{2}{c}{GWASH}
& \multicolumn{2}{c|}{LMM-REML}
& \multicolumn{2}{c}{GWASH}
& \multicolumn{2}{c}{LMM-REML} \\
\cmidrule(lr){2-3} \cmidrule(lr){4-5}
\cmidrule(lr){6-7} \cmidrule(lr){8-9}
& Formula-based & Empirical
& Formula-based & Empirical
& Formula-based & Empirical
& Formula-based & Empirical \\
\midrule
\multicolumn{9}{l}{Two-block exchangeable: $\rho_1 = 0.1,\ \rho_2 = 0.9$} \\
0   & 0.052 & 0.045 & 0.074 & 0.067 & 0.085 & 0.043 & 0.030 & 0.028 \\
1   & 0.036 & 0.030 & 0.056 & 0.057 & 0.039 & 0.024 & 0.027 & 0.023 \\
2   & 0.056 & 0.055 & 0.058 & 0.051 & 0.036 & 0.035 & 0.027 & 0.022 \\
10  & 0.056 & 0.055 & 0.058 & 0.051 & 0.036 & 0.035 & 0.026 & 0.021 \\
50  & 0.055 & 0.053 & 0.058 & 0.049 & 0.035 & 0.034 & 0.026 & 0.022 \\
100 & 0.054 & 0.053 & 0.057 & 0.048 & 0.034 & 0.035 & 0.025 & 0.022 \\
\bottomrule
\end{tabular}}
\end{table}

\end{document}